\newif\ifreview
\reviewtrue
\documentclass[11pt]{article}

\usepackage[top=1.25in, bottom=1.25in, left=1.25in, right=1.25in]{geometry}

\usepackage[utf8]{inputenc} 
\usepackage[T1]{fontenc}    
\usepackage{hyperref}       
\usepackage{url}            
\usepackage{booktabs}       
\usepackage{amsfonts}       
\usepackage{nicefrac}       
\usepackage{microtype}      
\usepackage{paralist}		
\usepackage{amsthm,amsmath,amssymb}
\usepackage{color,graphicx}
\usepackage{multicol}
\usepackage{multirow}
\usepackage{caption}
\usepackage{subcaption}
\usepackage[boxruled,linesnumbered]{algorithm2e}
\usepackage{algorithmic}
\usepackage{enumitem}
\usepackage{url}

\newtheorem{fact}{Fact}[section]
\newtheorem{claim}{Claim}[section]
\newtheorem{remark}{Remark}[section]

\newtheorem{theorem}{Theorem}[section]
\newtheorem{lemma}{Lemma}[section]

\newtheorem{definition}{Definition}[section]
\newtheorem{proposition}{Proposition}[section]

\renewcommand{\epsilon}{\varepsilon}
\newcommand{\eps}{\varepsilon}
\newcommand{\eat}[1]{}

\newcommand{\R}{\mathbb{R}}

\newcommand{\Z}{\mathbb{Z}}

\newcommand{\poly}{\operatorname{poly}}

\newcommand{\OPT}{\ensuremath{\mathsf{OPT}}\xspace}
\newcommand{\APX}{\ensuremath{\mathsf{APX}}\xspace}

\newcommand{\ignore}[1]{}

{\hspace*{\fill}$\Box$\par}

\newcommand{\kdist}{{\mathcal{K}}}

\newcommand{\calB}{{\mathcal{B}}}

\newcommand{\calI}{{\mathcal{I}}}
\newcommand{\calJ}{{\mathcal{J}}}

\newcommand{\calP}{{\mathcal{P}}}

\newcommand{\calZ}{{\mathcal{Z}}}

\title{\bf  Coresets for
	Clustering with Fairness Constraints}
\author{Lingxiao Huang\thanks{EPFL, Switzerland. Email: huanglingxiao1990@126.com} \and Shaofeng H.-C. Jiang\thanks{Weizmann Institute of Science, Israel. Email: shaofeng.jiang@weizmann.ac.il} \and Nisheeth K. Vishnoi\thanks{Yale University, USA. Email: nisheeth.vishnoi@yale.edu}}

\allowdisplaybreaks[4]

\begin{document}

	\maketitle

	\begin{abstract}
In a recent work, \cite{chierichetti2017fair} studied the following ``fair'' variants of classical clustering problems such as $k$-means and $k$-median: given a set of $n$ data points in $\mathbb{R}^d$ and a binary type associated to each data point, the goal is to cluster the points while ensuring that the proportion of each type in each cluster is roughly the same as its underlying proportion.
Subsequent work has focused on either extending this setting to when each data point has multiple, non-disjoint sensitive types such as race and gender \cite{bera2019fair}, or to address the problem that the clustering algorithms in the above work do not scale well ~\cite{schmidt2018fair,bercea2018cost,backurs2019scalable}.
The main contribution of this paper is an approach to clustering with fairness constraints that involve {\em multiple, non-disjoint} types, that is {\em also scalable}.
Our approach is based on novel constructions of coresets: for the $k$-median objective, we construct an $\eps$-coreset of size $O(\Gamma k^2 \eps^{-d})$ where $\Gamma$ is the number of distinct collections of groups that a point may belong to, and
for the $k$-means objective, we show how to construct an $\eps$-coreset of size $O(\Gamma k^3\eps^{-d-1})$.
The former  result is the first known coreset construction for the fair clustering problem with the $k$-median objective, and the latter result removes the dependence on the size of the full dataset as in~\cite{schmidt2018fair} and generalizes it to multiple, non-disjoint types.
Plugging our coresets into existing algorithms for fair clustering such as \cite{backurs2019scalable} results in the fastest algorithms for several cases.
Empirically, we assess our approach over the \textbf{Adult}, \textbf{Bank}, \textbf{Diabetes} and \textbf{Athlete} dataset, and show that the coreset sizes are much smaller than the full dataset; applying coresets indeed accelerates the running time of computing the fair clustering objective while ensuring that the resulting objective difference is small.
We also achieve a speed-up to recent fair clustering algorithms~\cite{backurs2019scalable,bera2019fair} by incorporating our coreset construction.
	\end{abstract}
	
	\newpage
	
	\tableofcontents
	
	\newpage

	\section{Introduction}
	\label{sec:introduction}
	
	Clustering algorithms are widely used in automated decision-making tasks, e.g., unsupervised learning~\cite{tan2006cluster}, feature engineering~\cite{jiang2008clustering,glassman2014feature}, and recommendation systems~\cite{burke2011recommender,pham2011clustering,das2014clustering}.
With the increasing applications of clustering algorithms in human-centric contexts, there is a growing concern 	that, if left unchecked, they can lead to discriminatory outcomes for protected groups, e.g., females/black people.
	For instance,  the proportion of a minority group assigned to some cluster can be far from its underlying proportion, even if clustering algorithms do not take the sensitive attribute into its decision making~\cite{chierichetti2017fair}.
Such an outcome may, in turn, lead to unfair treatment of minority groups, e.g., women may receive proportionally fewer job recommendations with high salary~\cite{datta2015automated,miller2015can} due to their underrepresentation  in the cluster of high salary recommendations.

	To address this issue, Chierichetti et al.~\cite{chierichetti2017fair} recently proposed the fair clustering problem that requires the clustering assignment to be \emph{balanced} with respect to a binary sensitive type, e.g., sex.\footnote{A type consists of several disjoint groups, e.g., the sex type consists of females and males.}
	Given a set $X$ of $n$ data points in $\R^d$ and a binary type associated to each data point, the goal is to cluster the points such that  the proportion of each type in each cluster is roughly the same as its underlying proportion, while ensuring  that the clustering objective is minimized.
	Subsequent work has focused on either extending this setting to when each data point has multiple, non-disjoint sensitive types~\cite{bera2019fair} (Definition~\ref{def:fair_clustering}), or to address the problem that the clustering algorithms do not scale well ~\cite{chierichetti2017fair,rosner2018privacy,schmidt2018fair,bercea2018cost,backurs2019scalable}.

	Due to the large scale of datasets, several existing fair clustering algorithms have to take samples instead of using the full dataset, since their running time is at least quadratic in the input size~\cite{chierichetti2017fair,rosner2018privacy,bercea2018cost,bera2019fair}.
Very recently, 	Backurs et al.~\cite{backurs2019scalable} propose a nearly linear approximation algorithm for fair $k$-median, but it only works for a binary type.
It is still unknown whether there exists a scalable approximation algorithm for multiple sensitive types~\cite{backurs2019scalable}.
To improve the running time of fair clustering algorithms, a powerful technique called coreset was introduced.
Roughly, a coreset for fair clustering is a	 small weighted point set, such that for any $k$-subset and any fairness constraint, the fair clustering objective computed over the coreset is approximately the same as that computed from the full dataset (Definition~\ref{def:fair_coreset}).
Thus, a coreset can be used as a proxy for the full dataset -- one can apply any fair clustering algorithm on the coreset, achieve a good approximate solution on the full dataset, and hope to speed up the algorithm.
As mentioned in~\cite{backurs2019scalable}, using coresets can indeed accelerate the computation time and save storage space for fair clustering problems.
Another benefit is that one may want to compare the clustering performance under different fairness constraints, and hence it may be  more efficient to repeatedly use coresets.
Currently, the only known result for coresets for fair clustering is by Schmidt et al.~\cite{schmidt2018fair}, who constructed an $\eps$-coreset for fair $k$-means clustering.
However, their coreset size includes a $\log n$ factor and only restricts to a sensitive type.
Moreover, there is no known coreset construction for other commonly-used clusterings, e.g., fair $k$-median.

\paragraph{Our contributions.}

	The main contribution of this paper is the efficient construction of coresets for clustering with fairness constraints that involve multiple, non-disjoint types.
	Technically, we show an efficient construction of $\eps$-coresets of size independent of $n$ for both fair $k$-median and fair $k$-means, summarized in Table~\ref{tab:coreset}.
	Let $\Gamma$ denote the number of distinct collections of groups that a point may belong to (see the first paragraph of Section~\ref{sec:kmedian} for the formal definition).
	\begin{compactitem}
		\item Our coreset for fair $k$-median is of size $O(\Gamma k^2 \eps^{-d})$ (Theorem~\ref{thm:kmedian}), which is the first known coreset to the best of our knowledge.
		\item For fair $k$-means, our coreset is of size $O(\Gamma k^3 \eps^{-d-1})$ (Theorem~\ref{thm:kmeans}), which improves the result of~\cite{schmidt2018fair} by an $\Theta(\frac{\log n}{\eps k^2})$ factor and generalizes it to multiple, non-disjoint types.
		\item As mentioned in~\cite{backurs2019scalable}, applying coresets can accelerate the running time of fair clustering algorithms, while suffering only an additional $(1+\eps)$ factor in the approxiation ratio.
		Setting $\eps=\Omega(1)$ and plugging our coresets into existing algorithms~\cite{schmidt2018fair,bera2019fair,backurs2019scalable}, we directly achieve scalable fair clustering algorithms, summarized in Table~\ref{tab:result}.
	\end{compactitem}
	We present novel technical ideas to deal with fairness constraints for coresets.
	\begin{compactitem}
		\item Our first technical contribution is a reduction to the case $\Gamma = 1$ (Theorem~\ref{thm:reduction}) which greatly simplifies the problem. Our reduction not only works for our specific construction, but also for all coreset constructions in general.
		\item Furthermore, to deal with the $\Gamma = 1$ case, we provide several interesting geometric observations for the optimal fair $k$-median/means clustering (Lemma~\ref{lm:kmedian_num_interval}), which may be of  independent interest.
	\end{compactitem}
	We implement our algorithm and conduct experiments on \textbf{Adult}, \textbf{Bank}, \textbf{Diabetes} and \textbf{Athlete} datasets.
	\begin{compactitem}
		\item A vanilla implementation results in a coreset with size that depends on $\eps^{-d}$.
		Our implementation is inspired by our theoretical results and produces coresets whose size is much smaller in practice.
		This improved implementation is still within the framework of our analysis, and the same worst case theoretical bound still holds.
		\item To validate the performance of our implementation, we experiment with varying $\eps$ for both fair $k$-median and $k$-means.
		As expected, the empirical error is well under the theoretical guarantee $\eps$, and the size does not suffer from the $\eps^{-d}$ factor.
		Specifically, for fair $k$-median, we achieve 5\% empirical error using only 3\% points of the original data sets, and we achieve similar error using 20\% points of the original data set for the $k$-means case.
		In addition, our coreset for fair $k$-means is better than uniform sampling and that of~\cite{schmidt2018fair} in the empirical error.
		\item The small size of the coreset translates to
		more than 200x speed-up (with error \textasciitilde 10\%) in the running time of computing the fair clustering objective when the fair constraint $F$ is given.
		{
		We also apply our coreset on the recent fair clustering algorithm~\cite{backurs2019scalable,bera2019fair}, and drastically improve the running time of the algorithm by approximately 5-15 times to~\cite{backurs2019scalable} and 15-30 times to~\cite{bera2019fair} for all above-mentioned datasets plus a large dataset \textbf{Census1990} that consists of 2.5 million records, even taking the coreset construction time into consideration.
		}
	\end{compactitem}

\begin{table}[t]
	\centering
	\captionsetup{font=small}
	\small
	\caption{Summary of coreset results.
		$T_1(n)$ and $T_2(n)$ denote the running time of an $O(1)$-approximate algorithm for $k$-median/means, respectively.
	}
	\label{tab:coreset}
	
	\begin{tabular}{ccc|cc}
		\toprule
		\multirow{2}{*}{} & \multicolumn{2}{c|}{$k$-Median} & \multicolumn{2}{c}{$k$-Means} \\
		& size  & construction time 	& size & construction time \\
		\midrule
		\cite{schmidt2018fair} & & &$O(\Gamma k\eps^{-d-2}\log n )$ & $\tilde{O}( k \eps^{-d-2}n \log n +T_2(n))$ \\
		This     & $O(\Gamma k^2\eps^{-d})$ & $O(k\eps^{-d+1} n +T_1(n))$ & $O(\Gamma k^3\eps^{-d-1})$ & $O( k\eps^{-d+1} n +T_2(n))$  \\
		\bottomrule
	\end{tabular}
\end{table}

\begin{table}[t]
	\centering
	\captionsetup{font=small}
	\small
	\caption{Summary of fair clustering algorithms.
		$\Delta$ denotes the maximum number of groups that a point may belong to,  and
		``multi'' means the algorithm can handle multiple non-disjoint types.
	}
	\label{tab:result}
	
	\begin{tabular}{cccc|ccc}
		\toprule
		\multirow{2}{*}{} & \multicolumn{3}{c|}{$k$-Median} & \multicolumn{3}{c}{$k$-Means} \\
		& multi  & approx. ratio &time 	& multi  & approx. ratio & time \\
		\midrule
		\cite{chierichetti2017fair} &  & O(1) & $\Omega(n^2)$ &  &  &   \\
		\cite{schmidt2018fair} &  &  &  &  & $O(1)$  & $n^{O(k)}$  \\
		\cite{backurs2019scalable}  &  & $\tilde{O}(d \log n)$ & $O(dn \log n+T_1(n))$ &  &  &   \\
		\cite{bercea2018cost}   &  & $(3.488,1)$ & $\Omega(n^2)$ &  & $(4.675,1)$ & $\Omega(n^2)$  \\
		\cite{bera2019fair}   & $\checkmark$ & $(O(1),4\Delta+4)$ & $\Omega(n^2)$ & $\checkmark$ & $(O(1),4\Delta+4)$ & $\Omega(n^2)$  \\
		 This    &  & $\tilde{O}(d \log n)$ & $O(dlk^2 \log (lk)+T_1(lk^2))$ &  & $O(1)$ & $(l k)^{O(k)}$  \\
		 This     & $\checkmark$ & $(O(1),4\Delta+4)$ & $\Omega(l^{2\Delta}k^4)$ & $\checkmark$ & $(O(1),4\Delta+4)$ & $\Omega(l^{2\Delta}k^6)$  \\
		\bottomrule
	\end{tabular}

\end{table}

\subsection{Other related works}
\label{subsec:related_work}

There are increasingly more works on fair
clustering algorithms.
Chierichetti et al.~\cite{chierichetti2017fair} introduced the fair clustering problem for a binary type and obtained approximation algorithms for fair $k$-median/center.
Backurs et al.~\cite{backurs2019scalable} improved the running time to nearly linear for fair $k$-median, but the approximation ratio is $\tilde{O}(d \log n)$.
R{\"o}sner and Schmidt~\cite{rosner2018privacy} designed a 14-approximate algorithm for fair $k$-center, and the ratio is improved to 5 by~\cite{bercea2018cost}.
For fair $k$-means, Schmidt et al.~\cite{schmidt2018fair} introduced the notion of fair coresets, and presented an efficient streaming algorithm.
More generally, Bercea et al.~\cite{bercea2018cost} proposed a bi-criteria approximation for fair $k$-median/means/center/supplier/facility location.
Very recently, Bera et al.~\cite{bera2019fair} presented a bi-criteria approximation algorithm for fair $(k,z)$-clustering problem (Definition~\ref{def:fair_clustering}) with arbitrary group structures (potentially overlapping), and Anagnostopoulos et al.~\cite{anagnostopoulos2019principal} improved their results by proposing the first constant-factor approximation algorithm.
It is still open to design a near linear time $O(1)$-approximate algorithm for the fair $(k,z)$-clustering problem.

There are other fair variants of clustering problems.
Ahmadian et al.~\cite{ahmadian2019clustering} studied a variant of the fair $k$-center problem in which the number of each type in each cluster has an upper bound, and proposed a bi-criteria approximation algorithm.
Chen et al.~\cite{chen2019proportionally} studied the fair clustering problem in which any $n/k$ points are entitled to form their own cluster if there is another center closer in distance for all of them.
{
Kleindessner et al.~\cite{kleindessner2019fair} investigate the fair $k$-center problem in which each center has a type, and the selection of the $k$-subset is restricted to include a fixed amount of centers belonging to each type.
In another paper~\cite{kleindessner2019guarantees}, they developed fair variants of spectral clusterings (a heuristic $k$-means clustering framework) by incorporating the proportional fairness constraints proposed by~\cite{chierichetti2017fair}.
}

The notion of coreset was first proposed by Agarwal et al.~\cite{agarwal2004approximating}.
There has been a large body of work for unconstrained clustering problems in Euclidean spaces~\cite{agarwal2002exact,har2004clustering,chen2006k,harpeled2007smaller,langberg2010universal,FL11,feldman2013turning,braverman2016new}).
Apart from these, for the general $(k,z)$-clustering problem, Feldman and Langberg~\cite{FL11}
presented an $\eps$-coreset of size $\tilde{O}(dk\eps^{-2z})$ in $\tilde{O}(nk)$ time.
Huang et al.~\cite{huang2018epsilon} showed an $\eps$-coreset of size $\tilde{O}(\mathrm{ddim}(X) \cdot k^3\eps^{-2z})$, where $\mathrm{ddim}(X)$ is doubling dimension that measures the intrinsic dimensionality of a space.
For the special case of $k$-means,
Braverman et al.~\cite{braverman2016new} improved the size to $\tilde{O}(k\eps^{-2}\cdot \min\left\{k/\eps, d \right\} )$ by a dimension reduction approach.
Works such as~\cite{FL11} use importance sampling technique which avoid the size factor $\eps^{-d}$, but it is unknown if such approaches can be used in fair clustering.

		\section{Problem definition}
	\label{sec:def}

	Consider a set $X\subseteq \R^d$ of $n$ data points, an integer $k$ (number of clusters), and $l$ groups $P_1,\ldots,P_l\subseteq X$.
	An \emph{assignment constraint}, which was proposed by Schmidt et al.~\cite{schmidt2018fair}, is a $k\times l$ integer matrix $F$.
	A clustering $\mathcal{C} = \{C_1,\ldots,C_k\}$, which is a $k$-partitioning of $X$, is said to satisfy assignment constraint $F$ if
	\[
	\left|C_i\cap P_j\right| = F_{ij},~\forall i\in [k], j\in [l].
	\]
	For a $k$-subset $C =\{ c_1, \ldots, c_k\} \subseteq X$ (the center set) and $z \in \mathbb{R}_{>0}$, we define $\kdist_z(X, F, C)$ as the minimum value of $\sum_{i\in [k]} \sum_{x \in C_i}{d^{z}(x, c_i)}$ among all clustering $\mathcal{C}= \{C_1,\ldots,C_k\}$ that satisfies $F$, which we call the optimal fair $(k,z)$-clustering value.
	If there is no clustering satisfying $F$, $\kdist_z(X, F, C)$ is set to be infinity.
	The following is our notion of coresets for fair $(k, z)$-clustering.
	This generalizes the notion introduced in~\cite{schmidt2018fair} which only considers a partitioned group structure.

	\begin{definition}[\bf Coreset for fair clustering]
		\label{def:fair_coreset}
		Given a set $X\subseteq \R^d$ of $n$ points and $l$ groups $P_1,\ldots,P_l\subseteq X$, a weighted point set $S\subseteq \R^d$ with weight function $w: S\rightarrow \R_{>0}$ is an $\eps$-coreset for the fair $(k,z)$-clustering problem, if for each $k$-subset $C\subseteq \R^d$ and each assignment constraint $F\in \mathbb{Z}_{\geq 0}^{k\times l}$, it holds that
		$
		\kdist_z(S, F, C) \in (1\pm \eps) \cdot \kdist_z(X, F, C).
		$
	\end{definition}
	Since points in $S$ might receive fractional weights, we change the definition of $\kdist_z$ a little, so that in evaluating $\kdist_z(S, F, C)$, a point $x \in S$ may be partially assigned to more than one cluster and the total amount of assignments of $x$ equals $w(x)$.

	The currently most general notion of fairness in clustering was proposed by~\cite{bera2019fair}, which enforces both upper bounds and lower bounds of any group's proportion in a cluster.
	\begin{definition}[\bf $(\alpha,\beta)$-proportionally-fair]
		\label{def:fair_cluster}
		A clustering $\mathcal{C}=(C_1,\ldots,C_k)$ is $(\alpha,\beta)$-proportionally-fair ($\alpha,\beta\in [0,1]^l$), if for each cluster $C_i$ and $j\in [l]$, it holds that
		$
		\alpha_j \leq \frac{\left| C_i\cap P_j \right|}{|C_i|} \leq \beta_j.
		$
	\end{definition}
	\noindent
	The above definition directly implies for each cluster $C_i$ and any two groups $P_{j_1}, P_{j_2}\in [l]$,
	$
	\frac{\alpha_{j_1}}{\beta_{j_2}} \leq \frac{\left| C_i\cap P_{j_1}  \right|}{\left| C_i\cap P_{j_2} \right|} \leq \frac{\beta_{j_1}}{\alpha_{j_2}}.
	$
	In other words, the fraction of points belonging to groups $P_{j_1}, P_{j_2}$ in each cluster is bounded from both sides.
	Indeed, similar fairness constraints have been investigated by works on other fundamental algorithmic problems such as data summarization \cite{celis2018fair}, ranking \cite{celis2018ranking,yang2017measuring}, elections~\cite{celis2018multiwinner}, personalization \cite{celis2018an,celis2019algorithmic}, classification~\cite{celis2019classification}, and online advertising~\cite{celis2019online}.
	Naturally, Bera et al.~\cite{bera2019fair} also defined the fair clustering problem with respect to $(\alpha,\beta)$-proportionally-fairness as follows.
	
	\begin{definition}[\bf $(\alpha,\beta)$-proportionally-fair $(k,z)$-clustering]
		\label{def:fair_clustering}
		Given a set $X\subseteq \R^d$ of $n$ points, $l$ groups $P_1,\ldots,P_l\subseteq X$, and two vectors $\alpha,\beta\in [0,1]^l$, the objective of $(\alpha,\beta)$-proportionally-fair $(k,z)$-clustering is to find a $k$-subset $C=\left\{c_1,\ldots,c_k\right\} \in \R^d$ and $(\alpha,\beta)$-proportionally-fair clustering $\mathcal{C}=\{C_1,\ldots,C_k\}$, such that the objective function
		$
		\sum_{i\in [k]} \sum_{x \in C_i}{d^{z}(x, c_i)}
		$
		is minimized.
	\end{definition}

	\noindent
	Our notion of coresets is very general, and we relate our notion of coresets to the $(\alpha,\beta)$-proportionally-fair clustering problem, via the following observation, which is similar to Proposition 5 in~\cite{schmidt2018fair}.

\begin{proposition}
	\label{pro:collection}
	Given a $k$-subset $C$, the assignment restriction required by $(\alpha,\beta)$-proportionally-fairness can be modeled as a collection of assignment constraints.
\end{proposition}

\noindent
As a result, if a weighted set $S$ is an $\eps$-coreset satisfying Definition~\ref{def:fair_coreset}, then for any $\alpha,\beta\in [0,1]^l$, the $(\alpha,\beta)$-proportionally-fair $(k,z)$-clustering value computed from $S$ must be a $(1 \pm \eps)$-approximation of that computed from $X$.

\begin{remark}
	Definition~\ref{def:fair_cluster} enforces fairness by looking at the proportion of a group in each cluster.
	We can also consider another type of constraints over the number of group points in each cluster, defined as follows.
	\begin{definition}[$(\alpha,\beta)$-fair]
		\label{def:fair_num}
		We call a clustering $\mathcal{C} = \{C_1,\ldots,C_k\}$ $(\alpha,\beta)$-fair ($\alpha,\beta\in \mathbb{Z}_{\geq 0}^{k\times l}$), if for each cluster $C_i$ and each $j\in [l]$, we have
		$
		\alpha_{ij} \leq \left| C_i\cap P_j \right| \leq \beta_{ij}.
		$
	\end{definition}
For instance, the above definition can be applied if one only cares about the diversity and requires that each cluster should contain at least one element from each group, i.e., $|C_i\cap P_j|\geq 1$ for all $i,j$.
	We can similarly define the $(\alpha,\beta)$-fair $(k,z)$-clustering problem with respect to the above definition as in Definition~\ref{def:fair_clustering}, and
	Proposition~\ref{pro:collection} still holds in this case.
	Hence, an $\eps$-coreset for fair $(k,z)$-clustering also preserves the clustering objective of the $(\alpha,\beta)$-fair $(k,z)$-clustering problem.
\end{remark}

	\section{Technical overview}
	\label{sec:overview}
	
	We introduce novel techniques to tackle the assignment constraints.
	Recall that $\Gamma$ denotes the number of distinct collections of groups that a point may belong to.
	Our first technical contribution is a general reduction to the $\Gamma=1$ case which works for any coreset construction algorithm (Theorem~\ref{thm:reduction}).
	The idea is to divide $X$ into $\Gamma$ parts with respect to the groups that a point belongs to, and construct a fair coreset with parameter $\Gamma=1$ for each group. The observation is that the union of these coresets is a coreset for the original instance and $\Gamma$.

	Our coreset construction for the case $\Gamma=1$ is based on the framework of~\cite{harpeled2007smaller} in which unconstrained $k$-median/means coresets are provided.
We first introduce the framework of~\cite{harpeled2007smaller} briefly and then show the main technical difficulty of our work.
The main observation of~\cite{harpeled2007smaller} is that it suffices to deal with $X$ that lies on a line.
Specifically, they show that it suffices to construct at most $O(k\eps^{-d+1})$ lines, project $X$ to their closest lines and construct an $\eps/3$-coreset for each line.
The coreset for each line is then constructed by partitioning the line into $\poly(k/\eps)$ contiguous sub-intervals, and designate at most two points to represent each sub-interval and include these points in the coreset.
	In their analysis, a crucially used property is that the clustering for any given centers partitions $X$ into $k$ contiguous parts on the line, since each point must be assigned to its nearest center.
However, this property might not hold in fair clustering, which is the main difficulty.
	Nonetheless, we manage to show a new structural lemma, that the optimal fair $k$-median/means clustering partitions $X$ into $O(k)$ contiguous intervals.
	For fair $k$-median, the key geometric observation is that there always exists a center whose corresponding optimal fair $k$-median cluster forms a contiguous interval (Claim~\ref{claim:single_interval}), and this combined with an induction implies the optimal fair clustering partitions $X$ into $2k-1$ intervals.
	For fair $k$-means, we show that each optimal fair cluster actually forms a single contiguous interval.
	Thanks to the new structural properties, plugging in a slightly different set of parameters in~\cite{harpeled2007smaller} yields fair coresets.

	\section{Coresets for fair $k$-median clustering}
\label{sec:kmedian}

In this section, we construct coresets for fair $k$-median ($z=1$).
For each $x\in X$, denote $\calP_x=\left\{i\in [l]: x\in P_i \right\}$ as the collection of groups that $x$ belongs to.
Let $\Gamma$ denote the number of distinct $\calP_x$'s.
Let $T_z(n)$ denote the running time of a constant approximation algorithm for the $(k,z)$-clustering problem.
The main theorem is as follows.

\begin{theorem}[\bf Coreset for fair $k$-median]
	\label{thm:kmedian}
	There exists an algorithm that constructs an $\eps$-coreset for the fair $k$-median problem of size $O(\Gamma k^2 \eps^{-d})$, in $O( k\eps^{-d+1} n +T_1(n))$ time.
\end{theorem}

\noindent
Note that $\Gamma$ is usually small.
For instance, if there is only a sensitive attribute~\cite{schmidt2018fair}, then each $\calP_x$ is a singleton and $\Gamma = l$.
More generally, let $\Lambda$ denote the maximum number of groups that any point belongs to, then $\Gamma\leq l^\Lambda$, but there is only $O(1)$ sensitive attributes for each point.

The main technical difficulty for the coreset construction is to deal with the assignment constraints.
We make an important observation (Theorem~\ref{thm:reduction}), that one only needs to prove Theorem~\ref{thm:kmedian} for the case $l=1$, and we thus focus on the case $l = 1$.
This theorem is a generalization of Theorem 7 in~\cite{schmidt2018fair}, and the coreset of~\cite{schmidt2018fair} actually extends to arbitrary group structure thanks to our theorem.
\begin{theorem}[\bf Reduction from $l$ groups to a single group]
	\label{thm:reduction}
	Suppose there exists an algorithm that computes an $\eps$-coreset of size $t$ for the fair $(k,z)$-clustering problem of $\widehat{X}$ satisfying that $l=1$, in time $T(|\widehat{X}|,\eps,k,z)$.
	There exists an algorithm, that given a set $X$ that can be partitioned into $\Gamma$ distinct subsets $X^{(1)},\ldots,X^{(\Gamma)}$ in which all points $x\in X^{(i)}$ correspond to the same collection $\calP_x$ for each $i\in [\Gamma]$, computes an $\eps$-coreset for the fair $(k,z)$-clustering problem of size $\Gamma t$, in time $\sum_{i\in [\Gamma]}T(|X^{(i)}|,\eps,k,z)$.
\end{theorem}

\begin{proof}
	Consider the case that $\Gamma=1$ in which all $\calP_x$'s are the same.
	Hence, this case can b reduced degenerated to $l=1$ and has an $\eps$-coreset of size $t$ by assumption.
	For each $i\in [\Gamma]$, suppose $S^{(i)}$ is an $\eps$-coreset for the fair $(k,z)$-clustering problem of $X^{(i)}$ where each point in $S^{(i)}$ belongs to all groups in $\calP_i$.
	Let $S:=\bigcup_{i\in [l]} S^{(i)}$.
	It is sufficient to prove $S$ is an $\eps$-coreset for the fair $(k,z)$-clustering problem of $X$, for both the correctness and the running time.

	Given a $k$-subset $C\subseteq \R^d$ and an assignment constraint $F$, let $C^\star_1,\ldots,C^\star_k$ be the optimal fair clustering of the instance $(X,F,C)$.
	Then for each collection $X^{(i)}$ ($i\in [\Gamma]$), we construct an assignment constraint $F^{(i)}\in \calZ^{k\times l}$ as follows: for each $j_1\in [k]$ and $j_2\in [l]$, let $F^{(i)}_{j_1,j_2}=0$ if $j_2\notin \calP_i$ and $\left|C^\star_{j_1}\cap X^{(i)}\right|$ if $j_2\in \calP_i$, i.e., $F^{(i)}_{j_1,j_2}$ is the number of points within $X^{(i)}$ that belong to $C_{j_1}\cap P_{j_2}$.
	By definition, we have that for each $j_1\in [k]$ and $j_2\in [l]$,
	\begin{eqnarray}
		\label{eq:reduction}
		F_{j_1,j_2} = \sum_{i\in [\Gamma]} F^{(i)}_{j_1,j_2}.
	\end{eqnarray}
	Then
	\begin{align*}
		\kdist_z(X, F, C) = & \sum_{i\in [l]} \kdist_z(X^{(i)},F^{(i)},C) & (\text{Defns. of $\kdist_z$ and $F^{(i)}$})\\
		\geq & (1- \eps) \cdot \sum_{i\in [l]} \kdist_z(S^{(i)}, F^{(i)}, C) & (\text{Defn. of $S^{(i)}$}) \\
		\geq & (1- \eps) \cdot \kdist_z(S, F, C) & (\text{Optimality and Eq.~\eqref{eq:reduction}}).
	\end{align*}
	Similarly, we can prove that $\kdist_z(S, F, C)\geq (1-\eps)\kdist_z(X, F, C)$.
	It completes the proof.
\end{proof}

Our coreset construction for both fair $k$-median and $k$-means are similar to that in~\cite{harpeled2007smaller}, except we use a different set of parameters. At a high level, the algorithm reduces general instances to instances where data lie on a line, and it only remains to give a coreset for the line case.

\begin{remark}
\label{remark:kcenter}
Theorem~\ref{thm:reduction} can be applied to construct an $\eps$-coreset of size $O(\Gamma k \eps^{-d+1})$ for the fair $k$-center clustering problem, since Har-Peled's coreset result~\cite{har2004clustering} directly provides an $\eps$-coreset of size $O(k \eps^{-d+1})$ for the case of $l=1$.
\end{remark}

\subsection{The line case}
\label{subsec:kmedian_1D}
Since $l=1$, we describe $F$ as an integer vector in $\mathbb{Z}_{\geq 0}^k$.
For a weighted point set $S$ with weight $w: S\rightarrow \R_{\geq 0}$, we define the \emph{mean} of $S$ by $\overline{S}:=\frac{1}{|S|}\sum_{p\in S} w(p)\cdot p$ and the \emph{error} of $S$ by $\Delta(S):=\sum_{p\in S} w(p)\cdot d(p,\overline{S})$.
Denote $\OPT$ as the optimal value of the unconstrained $k$-median clustering.
Our construction is similar to~\cite{harpeled2007smaller}, summarized in Algorithm~\ref{alg:kmedian_1D}.
An illustration of Algorithm~\ref{alg:kmedian_1D} may be found in Figure~\ref{fig:kmedian_1D}.

\begin{algorithm}[ht]
	\caption{FairMedian-1D($X,k$)}
	\label{alg:kmedian_1D}
	\KwIn{$X=\left\{x_1,\ldots,x_n\right\}\subset \R^d$ lying on the real line where $x_1\leq \ldots \leq x_n$, an integer $k\in [n]$, a number $\OPT$ as the optimal value of $k$-median clustering.}
	\KwOut{an $\eps$-coreset $S$ of $X$ together with weights $w: S\rightarrow \R_{\geq 0}$.}
	Set a threshold $\xi$ satisfying that $\xi = \frac{\eps \cdot \OPT}{30k}$ \;
	Consider the points from $x_1$ to $x_n$ and group them into batches in a greedy way: each batch $B$ is a maximal point set satisfying that $\Delta(B)\leq \xi$\;
	Denote $\calB(X)$ as the collection of all batches. Let $S\leftarrow \bigcup_{B\in \calB(X)} \overline{B}$\;
	For each point $x=\overline{B}\in S$, $w(x)\leftarrow |B|$\;
	Return $(S,w)$\;
\end{algorithm}

\begin{figure}
	\begin{center}
	\includegraphics[width=\textwidth]{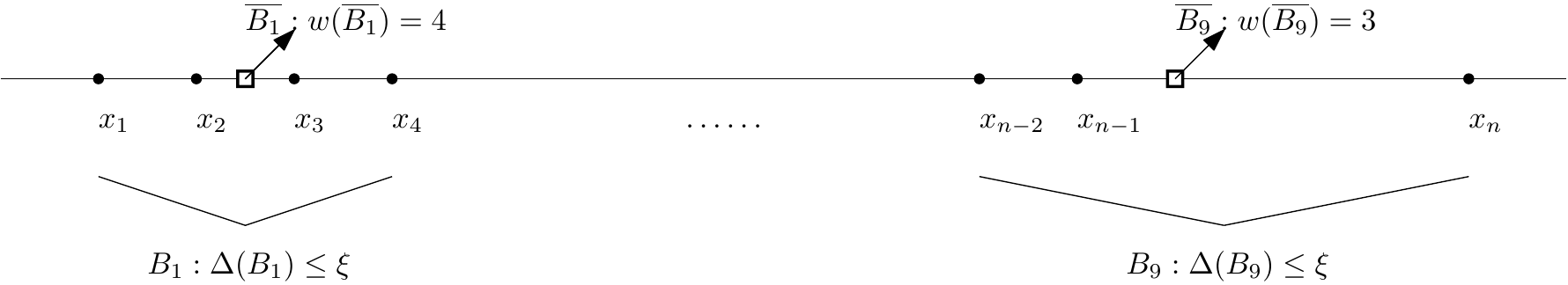}
	\end{center}
\caption{an illustration of Algorithm~\ref{alg:kmedian_1D} that divides $X$ into 9 batches.}
	\label{fig:kmedian_1D}
\end{figure}

\paragraph{Analysis.} We then prove the following theorem that shows the correctness of our coreset for the line case.

\begin{theorem}[\bf Coreset for fair $k$-median when $X$ lies on a line]
	\label{thm:kmedian_1D}
	Algorithm~\ref{alg:kmedian_1D} computes an $\eps/3$-coreset $S$ for fair $k$-median clustering of $X$, in time $O(|X|)$.
\end{theorem}

The running time is not hard since for each batch $B\in \calB(X)$, it only costs $O(|B|)$ time to compute $\overline{B}$.
Hence, Algorithm~\ref{alg:kmedian_1D} runs in $O(|X|)$ time.
In the following, we focus on correctness.
In~\cite{harpeled2007smaller}, it was shown that $S$ is an $\eps/3$-coreset for the unconstrained $k$-median clustering problem.
In their analysis, it is crucially used that the optimal clustering partitions $X$ into $k$ contiguous intervals.
Unfortunately, the nice ``contiguous'' property does not hold in our case because of the assignment constraint $F\in \R^k$.
To resolve this issue, we prove a new structural property (Lemma~\ref{lm:kmedian_num_interval}) that the optimal fair $k$-median clustering actually partitions $X$ into only $O(k)$ contiguous intervals.

\begin{lemma}[\bf Fair $k$-median clustering consists of $2k-1$ contiguous intervals]
	\label{lm:kmedian_num_interval}
	Suppose $X := \{x_1, \ldots, x_n\} \subset \R^d$ lies on the real line where $x_1 \leq \ldots \leq x_n$.
	For any $k$-subset $C = (c_1, \ldots, c_k)\in \R^d$ and any assignment constraints $F\in \Z_{\geq 0}^k$, there exists an optimal fair $k$-median clustering that partitions $X$ into at most $2k-1$ contiguous intervals.
\end{lemma}

\begin{proof}
	We prove by induction on $k$.
	The induction hypothesis is that, for any $k \geq 1$, Lemma~\ref{lm:kmedian_num_interval} holds for any data set $X$, any $k$-subset $C\subseteq \R^d$ and any assignment constraint $F\in \Z_{\geq 0}^k$.
	The base case $k=1$ holds trivially since all points in $X$ must be assigned to $c_1$.

	Assume the lemma holds for $k-1$ ($k\geq 2$) and we will prove the inductive step $k$.
	Let $C^\star_1, \ldots, C^\star_k$ be the optimal fair $k$-median clustering w.r.t. $C$ and $F$, where $C^\star_i \subseteq X$ is the subset assigned to center $c_i$.
	We present the structural property in Claim~\ref{claim:single_interval}, whose proof is given later.

	\begin{claim}
		\label{claim:single_interval}
		There exists $i \in [k]$ such that $C^\star_i$ consists of exactly one contiguous interval.
	\end{claim}
	We continue the proof of the inductive step by constructing a reduced instance $(X', F', C')$ where a) $C' := C \setminus \{c_{i_0}\}$; b) $X' = X \setminus C^\star_{i_0}$; c) $F'$ is formed by removing the $i_0$-th coordinate of $F$.
	Applying the hypothesis on $(X', F', C')$, we know the optimal fair $(k-1)$-median clustering consists of at most $2k-3$ contiguous intervals.
	Combining with $C^\star_{i_0}$ which has exactly one contiguous interval would increase the number of intervals by at most $2$.
	Thus, we conclude that the optimal fair $k$-median clustering for $(X, F, C)$ has at most $2k-1$ contiguous intervals.
	This finishes the inductive step.

	Finally, we complete the proof of Claim~\ref{claim:single_interval}.
	We first prove the following fact for preparation.
	\begin{fact}
			\label{fact:IDDI}
			Suppose $p, q \in \R^d$.
			Define $f : \R \to \R$ as $f(x) := d(x, p) - d(x, q)$ (here we abuse the notation by treating $x$ as a point in the $x$-axis of $\R^d$).
			Then $f$ is either ID or DI.\footnote{ID means that the function $f$ first (non-strictly) increases and then (non-strictly) decreases.
				DI means the other way round.}
	\end{fact}
		
	\begin{proof}
			Let $h_p$ and $h_q$ be the distance from $p$ and $q$ to the x-axis respectively, and let $u_p$ and $u_q$ be the corresponding $x$-coordinate of $p$ and $q$.
			We have
			\begin{align*}
			f(x) = \sqrt{(x - u_p)^2 + h_p^2} - \sqrt{(x - u_q)^2 + h_q^2}.
			\end{align*}
			Then we can regard $p,q$ as two points in $\R^2$ by letting $p=(u_p,h_p)$ and $q=(u_q,h_q)$.
			Also we have
			\begin{align*}
			f'(x)
			= \frac{x - u_p}{\sqrt{(x - u_p)^2 + h_p^2}} - \frac{x - u_q}{\sqrt{(x-u_q)^2 + h_q^2}}
			= \frac{x - u_p}{d(x, p)} - \frac{x - u_q}{d(x, q)}.
			\end{align*}
			W.l.o.g. assume that $u_p \leq u_q$.
			Next, we rewrite $f'(x)$ with respect to $\cos(\angle pxu_p)$ and $\cos(\angle qxu_q)$.

			\begin{enumerate}
				\item If $x \leq u_p$. Then $f'(x) = \frac{d(x, u_q)}{d(x, q)} - \frac{d(x, u_p)}{d(x, p)} = \cos(\angle qxu_q) - \cos(\angle pxu_p)$.
				\item If $u_p < x \leq u_q$. Then $f'(x) = \frac{d(x, u_p)}{d(x, p)} + \frac{d(x, u_q)}{d(x, q)} = \cos(\angle pxu_p) + \cos(\angle qxu_q)$.
				\item If $x > u_q$. Then $f'(x) = \frac{d(x, u_p)}{d(x, p)} - \frac{d(x, u_q)}{d(x, q)} = \cos(\angle pxu_p) - \cos(\angle qxu_q)$.
			\end{enumerate}
			
			\noindent
			Denote the intersecting point of line $pq$ and the $x$-axis to be $y$.
			Specificially, if $h_p=h_q$, we denote $y=-\infty$.
			Note that $f'(x)=0$ if and only if $x=y$.
			Now we analyze $f'(x)$ in two cases (whether or not $h_p \leq h_q$).
			\begin{itemize}
				\item Case i): $h_p \leq h_q$ which implies that $y<u_p$.
				When $x$ goes from $-\infty$ to $u_p$, first $f'(x) \leq 0$ and then $f'(x) \geq 0$.
				When $x > u_p$, $f'(x) \geq 0$.
				\item Case ii): $h_p > h_q$ which implies that $y>u_q$.
				When $x \leq u_q$, $f'(x) \geq 0$. When $x$ goes from $u_q$ to $+\infty$, first $f'(x) \geq 0$ and then $f'(x) \leq 0$.
			\end{itemize}
			Therefore, $f(x)$ is either DI or ID.
	\end{proof}
	
	\noindent	
	\textbf{Proof of Claim~\ref{claim:single_interval}.} Suppose for the contrary that for any $i\in [k]$, $C^\star_i$ consists of at least two contiguous intervals.
	Pick any $i$ and suppose $S_L, S_R \subseteq C^\star_i$ are two contiguous intervals such that $S_L$ lies on the left of $S_R$.
	Let $y_L$ denote the rightmost point of $S_L$ and $y_R$ denote the leftmost point of $S_R$.
	Since $S_L$ and $S_R$ are two distinct contiguous intervals, there exists some point $y \in X$ between $y_L$ and $y_R$ such that $y \in C^\star_j$ for some $j \neq i$.
	Define $g : \R \to \R$ as $g(x) := d(x, c_j) - d(x, c_i)$.
	By Fact~\ref{fact:IDDI}, we know that $g(x)$ is either ID or DI.

	If $g$ is ID, we swap the assignment of $y$ and $y_{\text{min}} := \arg\min_{x \in \{y_L, y_R\}}{g(x)}$ in the optimal fair $k$-median clustering.
	Since $g$ is ID, for any interval $P$ with endpoints $p$ and $q$, $\min_{x \in P}{g(x)} = \min_{x \in \{p, q\}}{g(x)}$.
	This fact together with $y_L \leq y \leq y_R$ implies that $g(y_{\text{min}}) - g(y)\leq 0$.
	Hence, the change of the objective is
	\[
		d(y, c_i) - d(y, c_j) - d(y_{\text{min}}, c_i) + d(y_{\text{min}}, c_j)
		= g(y_{\text{min}}) - g(y)
		\leq 0.
	\]
	This contradicts with the optimality of $C^\star$ and hence $g$ has to be DI.
		
	Next, we show that there is no $y' \in C^\star_j$ such that $y' < y_L$ or $y' > y_R$.
	We prove by contradiction and only focus on the case of $y' < y_L$,
	since the case of $z > y_R$ can be proved similarly by symmetry.
	We swap the assignment of $y_L$ and $y_{\text{max}} := \arg\max_{x \in \{ y, y' \}}{g(x)}$ in the optimal fair $k$-median clustering.
	The change of the objective is
	\begin{align*}
		&d(y_L, c_j) - d(y_L, c_i) - d(y_{\text{max}}, c_j) + d(y_{\text{max}}, c_i) \\
		=& g(y_L) - g(y_{\text{max}})
		\leq 0,
	\end{align*}
	where the last inequality is by the fact that $g$ is DI.
	This contradicts the optimality of $C^\star$.
	Hence, we conclude such $y'$ does not exist.

	Therefore, $\forall x \in C^\star_j$, $y_L < x < y_R$.
	By assumption, $C^\star_j$ consists of at least two contiguous intervals within $(y_L,y_R)$.
	However, we can actually do exactly the same argument for $C^\star_j$ as in the $i$ case, and eventually we would find a $j'$ such that $C^\star_{j'}$ lies inside a strict smaller interval $(y_L', y_R')$ of $X$, where $y_L < y_L' < y_R' < y_R$.
	Since $n$ is finite, we cannot do this procedure infinitely, which is a contradiction.
	This finishes the proof of Claim~\ref{claim:single_interval}.
\end{proof}

\subsection{Proof of Theorem~\ref{thm:kmedian_1D}.}
\label{subsec:proof_kmedian_1D}

Now we are ready to prove the main theorem of the last subsection.

\begin{proof}
	The proof idea is similar to that of Lemma 2.8 in~\cite{harpeled2007smaller}.
	We first rotate the space such that the line is on the $x$-axis and assume that $x_1\leq x_2\leq \ldots \leq x_n$.
	Given an assignment constraint $F\in \R^k$ and a $k$-subset $C=\left\{c_1,\ldots, c_k\right\}\subseteq \R^d$, let $c'_i$ denote the projection of point $c_i$ to the real line and assume that $c'_1\leq c'_2\leq \ldots\leq c'_k$.
	Our goal is to prove that
	\[
	\left|\kdist_1(S,F,C)-\kdist_1(X,F,C)\right| \leq \frac{\eps}{3}\cdot \kdist_1(X,F,C).
	\]

	\noindent
	By the construction of $S$, we build up a mapping $\pi: X\rightarrow S$ by letting $\pi(x) = \overline{B}$ for any $x\in B$.
	For each $i\in [k]$, let $C_i$ denote the collection of points assigned to $c_i$ in the optimal fair $k$-median clustering of $X$.
	By Lemma~\ref{lm:kmedian_num_interval}, $C_1,\ldots,C_k$ partition the line into at most $2k-1$ intervals $\calI_1, \ldots, \calI_t$ ($t\leq 2k-1$), such that all points of any interval $\calI_i$ are assigned to the same center.
	Denote an assignment function $f: X\rightarrow C$ by $f(x)=c_i$ if $x\in C_i$.
	Let $\widehat{\calB}$ denote the set of all batches $B$, which intersects with more than one intervals $\calI_i$, or alternatively, the interval $\calI(B)$ contains the projection of a center point of $C$ to the $x$-axis.
	Clearly, $|\widehat{\calB}|\leq 2k-2+k= 3k-2$.
	For each batch $B\in \widehat{\calB}$, we have
	\begin{align}
	\label{eq:kmedian1}
	\sum_{x\in B} d(\pi(x),f(x))-d(x,f(x)) \stackrel{\text{triangle ineq.}}{\leq} \sum_{x\in B} |d(x,\pi(x))| = \sum_{x\in B} |d(x,\overline{B})| \stackrel{\text{Defn. of $B$}}{\leq} \frac{\eps \OPT}{30k}.
	\end{align}
	Note that $X\setminus \bigcup_{B\in \widehat{\calB}} B$ can be partitioned into at most $3k-1$ contiguous intervals.
	Denote these intervals by $\calI'_1,\ldots,\calI'_{t'}$ ($t'\leq 3k-1$).
	By definition, all points of each interval $\calI'_i$ are assigned to the same center whose projection is outside $\calI'_i$.
	Then by the proof of Lemma 2.8 in~\cite{harpeled2007smaller}, we have that for each $\calI'_i$,
	\begin{align}
	\label{eq:kmedian2}
	\sum_{x\in \calI'_i} d(\pi(x),f(x))-d(x,f(x)) \leq 2\xi = \frac{\eps \OPT}{15k}.
	\end{align}
	Combining Inequalities~\eqref{eq:kmedian1} and~\eqref{eq:kmedian2}, we have
	\begin{eqnarray}
	\label{eq:kmedian3}
	\begin{split}
	& \kdist_1(S,F,C)-\kdist_1(X,F,C) \leq \sum_{x\in X} d(\pi(x),f(x))-d(x,f(x)) & (\text{Defn. of $\kdist_1(S,F,C)$}) \\
	= & \sum_{B\in \widehat{B}}  \sum_{x\in B} d(\pi(x),f(x))-d(x,f(x)) &\\
	&+ \sum_{i\in [t]} \sum_{x\in \calI'_i} d(\pi(x),f(x))-d(x,f(x)) &\\
	\leq & (3k-2)\cdot \frac{\eps \OPT}{30k} + (3k-1)\cdot \frac{\eps \OPT}{15k} & (\text{Ineqs.~~\eqref{eq:kmedian1} and~\eqref{eq:kmedian2}}) &\\
	\leq &\frac{\eps \OPT}{3} \leq \frac{\eps }{3}\cdot \kdist_1(X,F,C).&
	\end{split}
	\end{eqnarray}
	To prove the other direction, we can regard $S$ as a collection of $n$ unweighted points and consider the optimal fair $k$-median clustering of $S$.
	Again, the optimal fair $k$-median clustering of $S$ partitions the $x$-axis into at most $2k-1$ contiguous intervals, and can be described by an assignment function $f':S\rightarrow C$.
	Then we can build up a mapping $\pi':S\rightarrow X$ as the inverse function of $\pi$.
	For each batch $B$, let $S_B$ denote the collection of $|B|$ unweighted points located at $\overline{B}$.
	We have the following inequality that is similar to Inequality~\eqref{eq:kmedian1}
	\[
	\sum_{x\in S_B} d(\pi'(x),f'(x))-d(x,f'(x)) \leq \frac{\eps \OPT}{30k}.
	\]
	Suppose a contiguous interval $\calI$ consists of several batches and satisfies that all points of $\calI\cap S$ are assigned to the same center by $f'$ whose projection is outside $\calI$.
	Then by the proof of Lemma 2.8 in~\cite{harpeled2007smaller}, we have that
	\[
	\sum_{B\in \calI} 	\sum_{x\in S_B} d(\pi'(x),f'(x))-d(x,f'(x)) \leq 0.
	\]
	Then by a similar argument as for Inequality~\eqref{eq:kmedian3}, we can prove the other direction
	\[
	\kdist_1(X,F,C)-\kdist_1(S,F,C)  \leq \frac{\eps }{3}\cdot \kdist_1(X,F,C),
	\]
	which completes the proof.
\end{proof}

\subsection{Extending to higher dimension}
\label{subsec:kmedian_highdim}

The extension is the same as that of~\cite{harpeled2007smaller}.
For completeness, we describe the detailed procedure for coresets for fair $k$-median.
\begin{enumerate}
	\item We start with computing an approximate $k$-subset $C^\star=\left\{c_1,\ldots,c_k \right\}\subseteq \R^d$ such that $\OPT\leq \kdist_1(X,C^\star)\leq c\cdot \OPT$ for some constant $c>1$.\footnote{For example, we can set $c=10$ by~\cite{kanungo2004local}.}
	\item Then we partition the point set $X$ into sets $X_1,\ldots,X_k$ satisfying that $X_i$ is the collection of points closest to $c_i$.
	\item For each center $c_i$, we take a unit sphere centered at $c_i$ and construct an $\frac{\eps}{3c}$-net $N_{c_i}$\footnote{An $\eps$-net $Q$ means that for any point $p$ in the unit sphere, there exists a point $q\in Q$ satisfying that $d(p,q)\leq \eps$.} on this sphere.
	By Lemma 2.6 in~\cite{harpeled2007smaller}, $|N_{c_i}| = O(\eps^{-d+1})$ and may be computed in $O(\eps^{-d+1})$ time.
	Then for every $p\in N_{c_i}$, we emit a ray from $c_i$ to $p$.
	Overall, there are at most $O(k \eps^{-d+1})$ lines.
	\item For each $i\in [k]$, we project all points of $X_i$ onto the closest line around $c_i$.
	Let $\pi: X\rightarrow \R^d$ denote the projection function.
	By the definition of $\frac{\eps}{3c}$-net, we have that $\sum_{x\in X} d(x,\pi(x))\leq \eps \cdot \OPT/3$ which indicates that the projection cost is negligible.
	Then for each line, we compute an $\eps/3$-coreset of size $O(k \eps^{-1})$ for fair $k$-median by Theorem~\ref{thm:kmedian_1D}.
	Let $S$ denote the combination of coresets generated from all lines.
\end{enumerate}
\begin{proof}[Proof of Theorem~\ref{thm:kmedian}]
	Since there are at most $O(k \eps^{-d+1})$ lines and the coreset on each line is of size at most $O(k \eps^{-1})$ by Theorem~\ref{thm:kmedian_1D}, the total size of $S$ is $O(k^2 \eps^{-d})$.
	For the correctness, by the optimality of $\OPT$ (which is \emph{unconstrained} optimal),
	for any given assignment constraint $F\in \R^k$ and any $k$-subset $C\subseteq \R^d$, $\OPT\leq \kdist_1(X,F,C)$.
	Combining this fact with Theorem~\ref{thm:kmedian_1D}, we have that $S$ is an $\eps$-coreset for fair $k$-median clustering, by the same argument as in Theorem 2.9 of~\cite{harpeled2007smaller}.
	For the running time, we need $T_1(n)$ time to compute $C^\star$ and $\APX$ and	the remaining construction time is upper bounded by $O(k\eps^{-d+1} n)$ -- the projection process to lines.
	This completes the proof.
\end{proof}
\begin{remark}
	\label{remark:simplify_net}
	In fact, it suffices to emit a set of rays such that the total cost of projecting points to the rays is at most $\frac{\eps \cdot \OPT}{3}$.
	This observation is crucially used in our implementations (Section~\ref{sec:empirical}) to reduce the size of the coreset, particularly to avoid the construction of the $O(\eps)$-net which is of $O(\eps^{-d})$ size.
\end{remark}

	\section{Coresets for fair $k$-means clustering}
	\label{sec:kmeans}

	In this section, we show how to construct coresets for fair $k$-means.
	Similar to the fair $k$-median case, we apply the approach in~\cite{harpeled2007smaller}.
	The main theorem is as follows.

	\begin{theorem}[\bf Coreset for fair $k$-means]
		\label{thm:kmeans}
		There exists an algorithm that constructs $\eps$-coreset for the fair $k$-means problem of size $O(\Gamma k^3 \eps^{-d-1})$, in $O( k^2\eps^{-d+1} n +T_2(n,d,k))$ time.
	\end{theorem}

	\noindent
	Note that the above result improves the coreset size of~\cite{schmidt2018fair} by a $O(\frac{\log n}{\eps k^2})$ factor.
	Similar to the fair $k$-median case, it suffices to prove for the case $l=1$.
	Recall that an assignment constraint for $l=1$ can be described by a vector $F\in \R^k$.
	Denote $\OPT$ to be the optimal $k$-means value without any assignment constraint.

	\subsection{The line case}
	\label{subsec:kmeans_1D}
	
	Similar to~\cite{harpeled2007smaller}, we first consider the case that $X$ is a point set on the real line.
	Recall that for a weighted point set $S$ with weight $w: S\rightarrow \R_{\geq 0}$, the \emph{mean} of $S$ by $\overline{S}:=\frac{1}{|S|}\sum_{p\in S} w(p)\cdot p$, and the \emph{error} of $S$ by $\Delta(S):=\sum_{p\in S} w(p)\cdot d^2(p,\overline{S})$.
	Again, our construction is similar to~\cite{harpeled2007smaller}, summarized in Algorithm~\ref{alg:kmeans_1D}.
	The main difference to Algorithm~\ref{alg:kmedian_1D} is in Line 3: for each batch, we need to construct two weighted points for the coreset using a constructive lemma of \cite{harpeled2007smaller}, summarized in Lemma~\ref{lm:batch_property}.
    Also note that the selected threshold $\xi$ is different from that in Algorithm~\ref{alg:kmedian_1D}.

	\begin{algorithm}[ht]
		\caption{FairMeans-1D($X,k$)}
		\label{alg:kmeans_1D}
		\KwIn{$X=\left\{x_1,\ldots,x_n\right\}\subset \R^d$ lying on the real line where $x_1\leq \ldots \leq x_n$, an integer $k\in [n]$, a number $\OPT$ as the optimal value of $k$-means clustering.}
		\KwOut{an $\eps$-coreset $S$ of $X$ together with weights $w: S\rightarrow \R_{\geq 0}$.}
		Set a threshold $\xi$ satisfying that $\xi=  \frac{\eps^2 \OPT}{200k^2}$ \;
		Consider the points from $x_1$ to $x_n$ and group them into batches in a greedy way: each batch $B$ is a maximal point set satisfying that $\Delta(B)\leq \xi$\;
		Denote $\calB(X)$ as the collection of all batches. For each batch $B$, construct a collection $\calJ(B)$ of two points $q_1,q_2$ together with weights $w_1,w_2$ satisfying Lemma~\ref{lm:batch_property} \;
		Let $S\leftarrow \bigcup_{B\in \calB(X)} \calJ(B)$\;
		Return $(S,w)$\;
	\end{algorithm}

	\begin{lemma}[\bf Lemmas 3.2 and 3.4 in \cite{harpeled2007smaller}]
		\label{lm:batch_property}
		The number of batches is $O(k^2/\eps^2)$.
		For each batch $B$, there exist two weighted points $q_1,q_2\in \calI(B)$ together with weight $w_1,w_2$ satisfying that
		\begin{itemize}
			\item $w_1+w_2=|B|$.
			\item Let $\calJ(B)$ denote the collection of two weighted points $q_1$ and $q_2$.
			Then we have $\overline{\calJ(B)}=\overline{B}$ and $\Delta(B)=\Delta(\calJ(B))$.
			\item
			Given any point $q\in \R^d$, we have
			\[
			\kdist_2(B,q)=\Delta(B)+|B|\cdot d^2(q,\overline{B})=\kdist_2(\calJ(B),q).
			\]
		\end{itemize}
	\end{lemma}
	
	\paragraph{Analysis.} We argue that $S$ is indeed an $\eps/3$-coreset for the fair $k$-means clustering problem.
	By Theorem 3.5 in~\cite{harpeled2007smaller}, $S$ is an $\eps/3$-coreset for $k$-means clustering of $X$.
	However, we need to handle additional assignment constraints.
	To address this, we introduce the following lemma showing that every optimal cluster satisfying the given assignment constraint is within a contiguous interval.

	\begin{lemma}[\bf Clusters are contiguous for fair $k$-means]
		\label{lm:kmeans_contiguous}
		Suppose $X=\left\{x_1,\ldots,x_n\right\}$ where $x_1\leq x_2\leq \ldots \leq x_n$.
		Given an assignment constraint $F\in \R^k$ and a $k$-subset $C=\left\{c_1,\ldots, c_k\right\}\subseteq \R^d$.
		Then letting $C_i:=\left\{x_{1+\sum_{j<i} F_j},\ldots, x_{\sum_{j\leq i} F_j} \right\}$ ($i\in [k]$), we have
		\[
		\kdist_2(X,F,C) = \sum_{i\in [k]} \sum_{x\in C_i} d^2(x,c_i).
		\]
	\end{lemma}
	
	\begin{proof}
		Let $c'_i$ denote the projection of point $c_i$ to the real line and assume that $c'_1\leq c'_2\leq \ldots\leq c'_k$. We slightly abuse the notation by regarding point $c'_i$ as a real value.
		We prove the lemma by contradiction.
		Let $C_1,\ldots, C_k$ be the optimal fair clustering.
		By contradiction we assume that there exists $i_1<i_2$ and $j_1<j_2$ such that $x_{j_1}\in C_{i_2}$ and $x_{j_2}\in C_{i_1}$.
		By the definitions of $c'_{i_1}$ and $c'_{i_2}$, we have that
		\begin{eqnarray}
		\label{eq:lm1}
		d(c'_{i_1},x_{j_1})+d(c'_{i_2},x_{j_2}) \leq d(c'_{i_1},x_{j_2})+d(c'_{i_2},x_{j_1}),
		\end{eqnarray}
		and
		\begin{eqnarray}
		\label{eq:lm2}
		\max\left\{d(c'_{i_1},x_{j_1}), d(c'_{i_2},x_{j_2})\right\}\leq \max\left\{ d(c'_{i_1},x_{j_2}), d(c'_{i_2},x_{j_1})\right\}.
		\end{eqnarray}
		Combining Inequalities~\eqref{eq:lm1} and~\eqref{eq:lm2}, we argue that
		\begin{eqnarray}
		\label{eq:key}
		d^2(c'_{i_1},x_{j_1})+d^2(c'_{i_2},x_{j_2}) \leq d^2(c'_{i_1},x_{j_2})+d^2(c'_{i_2},x_{j_1})
		\end{eqnarray}
		by proving the following claim.
		\begin{claim}
			Suppose $a,b,c,d\geq 0$, $a+b\leq c+d$ and $a,b,c\leq d$.
			Then $a^2+b^2\leq c^2+d^2$.
		\end{claim}
		\begin{proof}
			If $a+b\leq d$, then we have $a^2+b^2\leq (a+b)^2\leq d^2\leq c^2+d^2$.
			So we assume that $a+b>d$.
			Let $e=a+b-d>0$.
			Since $a+b\leq c+d$, we have $e^2\leq c^2$.
			Hence, it suffices to prove that $a^2+b^2\leq e^2+d^2$.
			Note that
			\[
			e^2+d^2=(a+b-d)^2+d^2 = a^2+b^2+(d-a)(d-b) \geq a^2+b^2,
			\]
			which completes the proof.
		\end{proof}
		Now we come back to prove Lemma~\ref{lm:batch_property}.
		We have the following inequality.
		\begin{align*}
		& d^2(x_{j_1},c_{i_1})+d^2(x_{j_2},c_{i_2}) & \\
		= & d^2(x_{j_1},c'_{i_1})+d^2(c'_{i_1},c_{i_1})+d^2(x_{j_2},c'_{i_2})+d^2(c'_{i_2},c_{i_2}) & (\text{The Pythagorean theorem}) \\
		\leq & d^2(x_{j_1},c'_{i_2})+d^2(c'_{i_1},c_{i_1})+d^2(x_{j_2},c'_{i_1})+d^2(c'_{i_2},c_{i_2}) & (\text{Ineq.~\eqref{eq:key}}) \\
		= & d^2(x_{j_1},c_{i_2})+d^2(x_{j_2},c_{i_1}). & (\text{The Pythagorean theorem})
		\end{align*}
		It contradicts with the assumption that $x_{j_1}\in C_{i_2}$ and $x_{j_2}\in C_{i_1}$.
		Hence, we complete the proof.
	\end{proof}
	
	Now we are ready to give the following theorem.

\begin{theorem}[\bf Coreset for fair $k$-means when $X$ lies on a line]
	\label{thm:kmeans_1D}
	Algorithm~\ref{alg:kmeans_1D} outputs an $\eps/3$-coreset for fair $k$-means clustering of $X$ in time $O(|X|)$.
\end{theorem}
	
	\begin{proof}
		The proof is similar to that of Theorem 3.5 in~\cite{harpeled2007smaller}.
		The running time analysis is exactly the same.
		Hence, we only focus on the correctness analysis in the following.
		We first rotate the space such that the line is on the $x$-axis and assume that $x_1\leq x_2\leq \ldots \leq x_n$.
		Given an assignment constraint $F\in \R^k$ and a $k$-subset $C=\left\{c_1,\ldots, c_k\right\}\subseteq \R^d$, let $c'_i$ denote the projection of point $c_i$ to the real line and assume that $c'_1\leq c'_2\leq \ldots\leq c'_k$.
		Our goal is to prove that
		\[
		\left|\kdist_2(S,F,C)-\kdist_2(X,F,C)\right| \leq \frac{\eps}{3}\cdot \kdist_2(X,F,C).
		\]
		By Lemma~\ref{lm:kmeans_contiguous}, we have that the optimal fair clustering of $X$ should be
		\[
		C_i:=\left\{x_{1+\sum_{j<i} F_j},\ldots, x_{\sum_{j\leq i} F_j} \right\}
		\]
		for each $i\in [k]$.
		Hence, $\calI(C_1),\ldots,\calI(C_k)$ are disjoint intervals.
		Similarly, the optimal fair clustering of $X$ should be to scan weighted points in $S$ from left to right and cluster points of total weight $F_i$ to $c_i$.\footnote{Recall that a weighted point can be partially assigned to more than one cluster.}
		If a batch $B\in \calB(X)$ lies completely within some interval $\calI(C_i)$, then it does not contribute to the overall difference $\left|\kdist_2(S,F,C)-\kdist_2(X,F,C)\right|$ by Lemma~\ref{lm:batch_property}.

		Thus, the only problematic batches are those that contain an endpoint of $\calI(C_1),\ldots,\calI(C_k)$.
		There are at most $k-1$ such batches.
		Let $B$ be one such batch and $\calJ(B)=\left\{q_1,q_2\right\}$ be constructed as in Lemma~\ref{lm:batch_property}.
		For $i\in [k]$, let $V_i:= \calI(C_i)\cap B$.
		Let $T$ denote the collection of the $w_1$ left side points within $B$ and $T' = B\setminus T$.
		Note that $w_1$ may be fractional and hence $T$ may include a fractional point.
		Denote
		\[
		\eta:= \sum_{i\in [k]} \sum_{x\in V_i\cap T} d^2(x,q_1)+ \sum_{i\in [k]} \sum_{x\in V_i\cap T'} d^2(x,q_2).
		\]
		We have that
		\begin{eqnarray}
		\label{eq:thm_boundeta}
		\begin{split}
		\eta = & \sum_{i\in [k]} \sum_{x\in V_i\cap T} \left(d(x,\overline{B})-d(q_1,\overline{B} ) \right)^2+ \sum_{i\in [k]} \sum_{x\in V_i\cap T'} \left(d(x,\overline{B})-d(q_2,\overline{B} ) \right)^2 & \\
		\leq & \sum_{i\in [k]} \sum_{x\in V_i\cap T} \left(d^2(x,\overline{B})+ d^2(q_1,\overline{B} ) \right) +  \sum_{i\in [k]} \sum_{x\in V_i\cap T'} \left(d^2(x,\overline{B})+ d^2(q_2,\overline{B} ) \right) & \\
		= & \Delta(B) + \Delta(\calJ(B)) = 2\Delta(B) \quad \quad \quad \quad (\text{Lemma~\ref{lm:batch_property}}) &\\
		\leq & \frac{\eps^2 \OPT}{100 k} \quad \quad \quad \quad (\text{Construction of $B$}). &
		\end{split}
		\end{eqnarray}
		\noindent
		Then we can upper bound the contribution of $B$ to the overall difference $\left|\kdist_2(S,F,C)-\kdist_2(X,F,C)\right|$ by
		\begin{eqnarray}
		\label{eq:thm1}
		\begin{split}
		& \left|\sum_{i\in [k]} \sum_{x\in V_i\cap T} \left(d^2(x,c_i) - d^2(q_1,c_i) \right) + \sum_{i\in [k]} \sum_{x\in V_i\cap T'} \left(d^2(x,c_i) - d^2(q_2,c_i) \right) \right| & \\	
		\leq & 	\sum_{i\in [k]} \sum_{x\in V_i\cap T} \left|d^2(x,c_i) - d^2(q_1,c_i) \right| + \sum_{i\in [k]} \sum_{x\in V_i\cap T'} \left|d^2(x,c_i) - d^2(q_2,c_i) \right| & \\
		= & \sum_{i\in [k]} \sum_{x\in V_i\cap T} d(x,q_1)\left(d(x,c_i)+ d(q_1,c_i)\right)  + \sum_{i\in [k]} \sum_{x\in V_i\cap T'} d(x,q_2)\left(d(x,c_i) + d(q_2,c_i) \right)  & \\	
		\leq & 	\sum_{i\in [k]} \sum_{x\in V_i\cap T} d(x,q_1)\left(2d(x,c_i)+ d(x,q_1)\right)  + \sum_{i\in [k]} \sum_{x\in V_i\cap T'} d(x,q_2)\left(2d(x,c_i) + d(x,q_2) \right)  & \\
		= & \sum_{i\in [k]} \sum_{x\in V_i\cap T} d^2(x,q_1)+ \sum_{i\in [k]} \sum_{x\in V_i\cap T'} d^2(x,q_2) & \\
		& + 2\sum_{i\in [k]} \sum_{x\in V_i\cap T}  d(x,q_1) d(x,c_i)  + 2\sum_{i\in [k]} \sum_{x\in V_i\cap T'} d(x,q_2) d(x,c_i)   & \\
		\leq & \eta+ 2\sqrt{\eta}\sqrt{\sum_{i\in [k]} \sum_{x\in V_i}  d^2(x,c_i)} \quad \quad \quad \quad (\text{Defn. of $\eta$ and Cauchy-Schwarz}) &\\
		\leq & \frac{\eps^2 \OPT}{50k} + \frac{2\eps}{7k} \sqrt{\OPT \cdot \kdist_2(X,F,C)}\quad \quad \quad \quad   (\text{Ineq.~\eqref{eq:thm_boundeta}}) &\\
		\leq &  \frac{\eps^2 \OPT}{100k} + \frac{2\eps}{10k} \cdot \frac{\OPT+\sum_{i\in [k]} \sum_{x\in V_i}  d^2(x,c_i)}{2} & \\
		\leq & \frac{\eps \OPT}{5k} + \frac{\eps \sum_{i\in [k]} \sum_{x\in V_i}  d^2(x,c_i)}{10k}. &
		\end{split}
		\end{eqnarray}
		Since there are at most $k-1$ such batches, we conclude that the their total contribution to the error $\left|\kdist_2(S,F,C)-\kdist_2(X,F,C)\right|$ can be upper bounded by
		\[
		\frac{\eps \OPT}{5} + \frac{\eps \kdist_2(X,F,C) }{10k} \leq \frac{\eps}{3} \cdot \kdist_2(X,F,C).
		\]
		It completes the proof.
	\end{proof}
	
	\subsection{Extending to higher dimension}
	\label{subsec:kmeans_highdim}
	
	The extension is almost the same as fair $k$-median, except that we apply Theorem~\ref{thm:kmeans_1D} to construct the coreset on each line.
	Let $S$ denote the combination of coresets generated from all lines.

	\begin{proof}[Proof of Theorem~\ref{thm:kmeans}]
		By the above construction, the coreset size is $O(k^3 \eps^{-d-1})$.
		For the correctness, Theorem 3.6 in~\cite{harpeled2007smaller} applies an important fact that for any $k$-subset $C\subseteq \R^d$,
		\[
		\kdist_2(X,C^\star)\leq c\cdot\kdist_2(X,C).
		\]
		In our setting, we have a similar property.
		Note that for any given \texttt{assignment constraint} $F\in \R^k$ and any $k$-subset $C\subseteq \R^d$, we have
		\[
		\kdist_2(X,C^\star)\leq c\cdot\kdist_2(X,F,C).
		\]
		Then combining this fact with Theorem~\ref{thm:kmeans_1D}, we have that $S$ is an $\eps$-coreset for the fair $k$-means clustering problem, by the same argument as that of Theorem 3.6 in~\cite{harpeled2007smaller}.
	\end{proof}
\section{Empirical results}
\label{sec:empirical}
We implement our algorithm and evaluate its performance on real datasets.
The implementation mostly follows our description of algorithms, but
a vanilla implementation would bring in an $\eps^{-d}$ factor in the coreset size.
To avoid this, as observed in Remark~\ref{remark:simplify_net}, we may actually emit any set of rays as long as the total projection cost is bounded, instead of $\eps^{-d}$ rays.
We implement this idea by finding the smallest integer $m$ and $m$ lines, such that the minimum cost of projecting data onto $m$ lines is within the error threshold.
In our implementation for fair $k$-means, we adopt the widely used Lloyd's heuristic~\cite{lloyd1982least} to find the $m$ lines,
where the only change to Lloyd's heuristic is that, for each cluster, we need to find a \emph{line} that minimizes the projection cost instead of a point, and we use SVD to efficiently find this line optimally.
Unfortunately, the above approach does not work for fair $k$-median, as the SVD does not give the optimal line.
As a result, we still need to construct the $\eps$-net, but we alternatively employ some heuristics to find the net adaptively w.r.t. the dataset.

Our evaluation is conducted on four datasets: \textbf{Adult} (\textasciitilde 50k), \textbf{Bank} (\textasciitilde 45k), \textbf{Diabetes} (\textasciitilde 100k) and \textbf{Athlete} (\textasciitilde 200k)~\cite{chierichetti2017fair,schmidt2018fair,bera2019fair}.
For all datasets, we choose numerical features to form a vector in $\mathbb{R}^d$ for each record, where $d=6$ for \textbf{Adult}, $d=10$ for \textbf{Bank}, $d=29$ for \textbf{Diabetes} and $d=3$ for \textbf{Athlete}.
We use $\ell_2$ to measure the distance of these vectors.
We choose two sensitive types for the first three datasets: sex and marital for \textbf{Adult} (9 groups, $\Gamma = 14$); marital and default for \textbf{Bank} (7 groups, $\Gamma = 12$); sex and age for \textbf{Diabetes} (12 groups, $\Gamma = 20$), and we choose a binary sensitive type sex for \textbf{Athlete} (2 groups, $\Gamma = 2$).
In addition, in Section~\ref{sec:other_empirical}, we will also discuss how the following affects the result: a) choosing a binary type as the sensitive type, or b) normalization of the dataset.
We pick $k = 3$ (i.e. number of clusters) throughout our experiment.
We define the \emph{empirical error} as $| \frac{\mathcal{K}_z(S, F, C)}{\mathcal{K}_z(X, F, C)} - 1 |$ (which is the same measure as $\epsilon$) for some $F$ and $C$.
To evaluate the empirical error, we draw 500 independent random samples of $(F, C)$ and report the maximum empirical error among these samples.
For each $(F,C)$, the fair clustering objectives $\kdist_z(\cdot,F,C)$ may be formulated as integer linear programs (ILP).
We use \textbf{CPLEX}~\cite{ibm2015ibm} to solve the ILP's, report the average running time\footnote{The experiments are conducted on a 4-Core desktop CPU with 64 GB RAM.} $T_X$ and $T_S$ for evaluating the objective on dataset $X$ and coreset $S$ respectively, and also report the running time $T_C$ for constructing coreset $S$.

For both $k$-median and $k$-means, we employ \emph{uniform sampling} (\textbf{Uni}) as a baseline, in which we partition $X$ into $\Gamma$ parts according to distinct $\calP_x$'s (the collection of groups that $x$ belongs to) and take uniform samples from each collection.
Additionally, for $k$-means, we select another baseline from a recent work~\cite{schmidt2018fair} that presented a coreset construction for fair $k$-means, whose implementation is based on the \textbf{BICO} library which is a high-performance coreset-based library for computing k-means clustering~\cite{fichtenberger2013BICO}.
We evaluate the performance of our coreset for fair $k$-means against \textbf{BICO} and \textbf{Uni}.
As a remark of \textbf{BICO} and \textbf{Uni} implementations, they do not support specifying parameter $\eps$, but a hinted size of the resulted coreset.
Hence, we start with evaluating our coreset, and set the hinted size for \textbf{Uni} and \textbf{BICO} as the size of our coreset.

{
We also showcase the speed-up to two recently published approximation algorithms by applying a 0.5-coreset. 
The first algorithm is a practically efficient, $O(\log n)$-approximate algorithm for fair $k$-median~\cite{backurs2019scalable} that works for a binary type, referred to as \textbf{FairTree}.
The other one is a bi-criteria approximation algorithm~\cite{bera2019fair} for both fair $k$-median and $k$-means, referred to as \textbf{FairLP}.
We slightly modify the implementations of \textbf{FairTree} and \textbf{FairLP} to enable them work with our coreset, particularly making them handle weighted inputs efficiently.
We do experiments on a large dataset \textbf{Census1990} which consists of about 2.5 million records (where we select $d = 13$ features and a binary sensitive type), in addition to the above-mentioned \textbf{Adult}, \textbf{Bank}, \textbf{Diabetes} and \textbf{Athlete} datasets. 
}

\begin{table}[ht]
	\centering
	\captionsetup{font=small}
	\small
	\caption{performance of $\eps$-coresets for fair $k$-median w.r.t. varying $\eps$.
	}
	\label{tab:kmedian}
	\begin{tabular}{cccccccc}
		\toprule
		& \multirow{2}{*}{$\eps$} & \multicolumn{2}{c}{emp. err.} & \multirow{2}{*}{size} &  \multirow{2}{*}{$T_{S}$ (ms)} & \multirow{2}{*}{$T_{C}$ (ms)} & \multirow{2}{*}{$T_{X}$ (ms)} \\
		& &  Ours & \textbf{Uni} & & & &  \\
		\midrule
		\multirow{7}{*}{\rotatebox{90}{\textbf{Adult}}} &
		10\% & 2.36\% & 12.28\% & 262 & 13 &  408 & 7101 \\
		& 15\% & 1.96\% & 19.86\% & 210  & 11 & 318 & - \\
		& 20\% & 4.36\%  & 17.17\%  & 215  & 12  & 311   & - \\
		& 25\% & 5.48\% & 20.71\% &  180  & 10 & 283  & - \\
		& 30\% & 4.46\% & 15.12\% & 161  & 9 & 295    & - \\
		& 35\% & 6.37\% & 32.54\%  & 171 & 10  & 267    & - \\
		& 40\% & 8.52\%  & 31.96\%  & 139  & 9  & 282 & -  \\
		\midrule
		\multirow{7}{*}{\rotatebox{90}{\textbf{Bank}}} &
		10\% & 1.45\% & 5.32\% & 2393 & 111  & 971  & 5453 \\
		& 15\% & 2.17\% & 5.47\% & 1130 & 53 & 704  & - \\
		& 20\% & 2.24\%  & 3.38\%  & 1101 & 50  & 689   & - \\
		& 25\% & 3.39\% & 7.26\% &  534  & 25  & 525 & - \\
		& 30\% & 4.18\% & 14.60\% & 506  & 24  & 476    & - \\
		& 35\% & 7.29\% & 13.50\% & 512  & 24  & 517    & - \\
		& 40\% & 5.35\%  & 10.53\%  & 293    & 14  & 452  & -  \\
		\midrule
		\multirow{7}{*}{\rotatebox{90}{\textbf{Diabetes}}} &
		10\% & 0.55\% & 6.38\% & 85822  & 12112 &  141212 & 17532 \\
		& 15\% & 0.86\% & 14.56\% & 65093  & 8373 &  54155  & - \\
		& 20\% & 1.62\%  & 15.44\%  & 34271  & 3267  & 16040   & - \\
		& 25\% & 2.43\% & 7.62\% &  17155  & 1604 & 8071  & - \\
		& 30\% & 3.61\% & 1.92\% & 6693   & 411  & 5017  & - \\
		& 35\% & 4.31\% & 2.11\% & 4359  & 256 & 4063  & - \\
		& 40\% & 5.33\%  & 3.67\%  & 2949   & 160 & 3916  & -  \\
		\midrule
		\multirow{7}{*}{\rotatebox{90}{\textbf{Athlete}}} &
		10\% & 1.14\% & 2.87\% & 3959  & 96 & 8141  & 74851 \\
		& 15\% & 2.00\% & 1.50\% & 1547  & 38 & 5081  & - \\
		& 20\% & 2.59\%  & 4.38\%  & 685  & 19  & 3779  & - \\
		& 25\% & 3.83\% & 7.67\% & 439  & 13  & 3402  & - \\
		& 30\% & 4.86\% & 4.98\% & 316   & 11  & 2763   & - \\
		& 35\% & 6.31\% & 7.47\% & 160  & 8 & 2496  & - \\
		& 40\% & 8.25\%  & 16.59\%  & 112 & 7  & 2390 & -  \\
		\bottomrule
	\end{tabular}
	
\end{table}

\begin{table}[ht]
	\centering
	\captionsetup{font=small}
	\small
	\caption{performance of $\eps$-coresets for fair $k$-means w.r.t. varying $\eps$.
	}
	\label{tab:kmeans}
	
		\begin{tabular}{cccccccccc}
		\toprule
		& \multirow{2}{*}{$\eps$} & \multicolumn{3}{c}{emp. err.} & \multirow{2}{*}{size} &  \multirow{2}{*}{$T_{S}$ (ms)} & \multicolumn{2}{c}{$T_{C}$ (ms)} & \multirow{2}{*}{$T_{X}$ (ms)} \\
		& &  Ours  & \textbf{BICO} & \textbf{Uni} & &  & Ours  & \textbf{BICO} &  \\
		\midrule
		\multirow{7}{*}{\rotatebox{90}{\textbf{Adult}}} &
		10\% & 0.28\% & 1.04\% & 10.63\% & 880 & 44 & 1351 & 786 & 7404 \\
		& 15\% & 0.56\% & 2.14\% & 4.48\% & 714  & 36  & 561 & 755 &  - \\
		& 20\% & 0.55\% & 1.12\% & 2.87\% & 610 & 29 &  511 & 788 & - \\
		& 25\% & 1.37\% & 2.29\% & 17.90\% & 543 &  27 &  526 & 781  & - \\
		& 30\% & 1.17\% & 4.06\% & 19.91\% & 503 &  26 &  495 & 750  & - \\
		& 35\% & 1.63\% & 4.17\% & 29.85\% & 457 &  24 & 512 & 787 & - \\
		& 40\% & 2.20\% & 4.45\% & 48.10\% & 433 &  22 &  492 & 768  & - \\
		\midrule
		\multirow{7}{*}{\rotatebox{90}{\textbf{Bank}}} &
		10\% & 2.85\% & 2.71\% & 30.68\% & 409 & 19 &  507 & 718 & 5128 \\
		& 15\% & 2.93\% & 4.34\% & 25.44\% & 328  & 16 &  512 & 687 & - \\
		& 20\% & 2.93\% & 4.59\% & 45.09\% & 280 & 14  & 478 & 712  & - \\
		& 25\% & 2.61\% & 4.99\% & 20.35\% & 242 & 12  & 509 & 694  & - \\
		& 30\% & 2.68\% & 6.10\% & 24.82\% & 230 & 11  & 531 & 711  & - \\
		& 35\% & 2.41\% & 5.85\% & 36.48\% & 207 & 11 & 528 & 728  & - \\
		& 40\% & 2.30\% & 5.66\% & 33.42\% & 194 & 10  & 505 & 690 & - \\
		\midrule
		\multirow{7}{*}{\rotatebox{90}{\textbf{Diabetes}}} &
		10\% & 4.39\% & 10.54\% & 1.91\% & 50163 & 5300  & 65189 & 2615  & 16312 \\
		& 15\% & 7.99\% & 5.83\% & 10.74\% & 11371  & 772 & 11664 & 1759  & - \\
		& 20\% & 11.24\% & 11.32\% & 4.41\% & 3385  & 168  & 5138 & 1544  & - \\
		& 25\% & 14.91\% & 15.76\% & 7.87\% & 1402  & 65 & 2999 & 1491 & - \\
		& 30\% & 14.52\% & 20.54\% & 13.46\% & 958  & 44  & 2680 & 1480& - \\
		& 35\% & 14.20\% & 20.72\% & 11.52\% & 870 & 41 & 2594 & 1488  & - \\
		& 40\% & 13.95\% & 22.05\% & 10.92\% & 775 & 35 & 2657 & 1462  & - \\
		\midrule
		\multirow{7}{*}{\rotatebox{90}{\textbf{Athlete}}} &
		10\% & 5.43\% & 4.94\% & 10.96\% & 1516 & 36 & 14534 & 1160 & 73743 \\
		& 15\% & 7.77\% & 11.72\% & 10.08\% & 491 & 14 & 6663 & 1099 & - \\
		& 20\% & 11.41\% & 21.31\% & 10.62\% & 213 & 9 & 3566 & 1090  & - \\
		& 25\% & 13.25\% & 26.33\% & 11.48\% & 112 & 7  & 2599 & 1114  & - \\
		& 30\% & 13.18\% & 29.97\% & 16.93\% & 98  & 7  & 2591 & 1076 & - \\
		& 35\% & 12.41\% & 24.86\% & 27.60\% & 92 & 7 & 2606 & 1056  & - \\
		& 40\% & 13.01\% & 29.74\% & 152.31\% & 83 & 6 & 2613 & 1066 & - \\
		\bottomrule
	\end{tabular}
\end{table}

\begin{table}[ht]
	\centering
	\captionsetup{font=small}
	\small
	\caption{speed-up of fair clustering algorithms using our coreset. $\text{obj}_{\text{ALG}}/\text{obj}_{\text{ALG}}$ is the runtime/clustering objective w/o our coreset, $T_{\text{ALG}}'/\text{obj}'_{\text{ALG}}$ is the runtime/clustering objective on top of our coreset, and $T_{\text{C}}$ is time to construct coreset.}
	\label{tab:speed-up_backurs}
	
	\begin{tabular}{ccccccc}
		\toprule
		& ALG & $\text{obj}_{\text{ALG}}$ & $\text{obj}'_{\text{ALG}}$ & $T_{\text{ALG}}$ (s) & $T_{\text{ALG}}'$ (s) & $T_{\text{C}}$ (s) \\
		\midrule
		\multirow{2}{*}{\textbf{Adult}} &\textbf{FairTree} ($z=1$) & $2.09\times 10^9$ & $1.23\times 10^9$ & 12.62 & 0.38 & 0.63 \\
		& \textbf{FairLP} ($z=2$) & $1.23\times 10^{14}$ & $1.44\times 10^{14}$ & 19.92 & 0.20 & 1.03 \\
		\midrule
		\multirow{2}{*}{\textbf{Bank}} &\textbf{FairTree} ($z=1$) & $5.69\times 10^6$ & $4.70\times 10^6$ & 14.62 & 0.64 & 0.60 \\
		& \textbf{FairLP} ($z=2$) & $1.53\times 10^{9}$ & $1.46\times 10^{9}$ & 17.41 & 0.08 & 0.50 \\
		\midrule
		\multirow{2}{*}{\textbf{Diabetes}} &\textbf{FairTree} ($z=1$) & $1.13\times 10^6$ & $9.50\times 10^5$ & 19.26 & 1.70 & 2.96 \\
		& \textbf{FairLP} ($z=2$) & $1.47\times 10^{7}$ & $1.08\times 10^{7}$ & 55.11 & 0.41 & 2.61 \\
		\midrule
		\multirow{2}{*}{\textbf{Athlete}} &\textbf{FairTree} ($z=1$) & $2.50\times 10^6$ & $2.42\times 10^6$ & 29.94 & 1.34 & 2.35 \\
		& \textbf{FairLP} ($z=2$) & $3.33\times 10^{7}$ & $2.89\times 10^{7}$ & 37.50 & 0.03 & 2.42 \\
		\midrule
		\multirow{2}{*}{\textbf{Census1990}} &\textbf{FairTree} ($z=1$) & $9.38\times 10^6$ & $7.65\times 10^6$ & 450.79 & 23.36 & 20.28 \\
		& \textbf{FairLP} ($z=2$) & $4.19\times 10^{7}$ & $1.32\times 10^{7}$ & 1048.72 & 0.06 & 31.05 \\
		\bottomrule
	\end{tabular}
	
\end{table}

\subsection{Results}
\label{subsec:result}

Table~\ref{tab:kmedian} and~\ref{tab:kmeans} summarize the accuracy-size trade-off of our coresets for fair $k$-median and $k$-means respectively, under different error guarantee $\eps$.
Since the coreset construction time $T_C$ for \textbf{Uni} is very small (usually less than 50 ms) we do not report it in the table.
From the table, a key finding is that the size of the coreset does not suffer from the $\eps^{-d}$ factor thanks to our optimized implementation.
As for the fair $k$-median, the empirical error of our coreset is well under control.
In particular, to achieve ~5\% empirical error, only less than 3 percents of data is necessary for all datasets, and this results in a \textasciitilde 200x acceleration in evaluating the objective and 10x acceleration even taking the coreset construction time into consideration.\footnote{
The same coreset may be used for clustering with any assignment constraints, so its construction time would be averaged out if multiple fair clustering tasks are performed.
}
Regarding the running time, our coreset construction time scales roughly linearly with the size of the coreset, which means our algorithm is output-sensitive.
The empirical error of \textbf{Uni} is comparable to ours on \textbf{Diabetes}, but the worst-case error is unbounded (2x-10x to our coreset, even larger than $\eps$) in general and seems not stable when $\eps$ varies.

Our coreset works well for fair $k$-means, and it also offers significant acceleration of evaluating the objective.
Compared with \textbf{BICO}, our coreset achieves smaller empirical error for fixed $\eps$ and the construction time is between 0.5x to 2x that of \textbf{BICO}.
Again, the empirical error of \textbf{Uni} could be 2x smaller than ours and \textbf{BICO} on \textbf{Diabetes}, but the worst-case error is unbounded in general.

Table~\ref{tab:speed-up_backurs} demonstrates the speed-up to \textbf{FairTree} and \textbf{FairLP} with the help of our coreset.
We observed that the adaption of our coresets offers a 5x-15x speed-up to \textbf{FairTree} and a 15x-30x speed-up to \textbf{FairLP} for all datasets, even taking the coreset construction time into consideration.
Specifically, the runtime on top of our coreset for \textbf{FairLP} is less than 1s for all datasets, which is extremely fast.
We also observe that the clustering objective $\text{obj}'_\text{ALG}$ on top of our coresets is usually within 0.6-1.2 times of $\text{obj}_\text{ALG}$ which is the objective without the coreset (noting that coresets might shrink the objective). 
%
%
The only exception is \textbf{FairLP} on \textbf{Census1990}, in which $\text{obj}'_\text{ALG}$ is only 35\% of $\text{obj}_\text{ALG}$.
A possible reason is that in the implementation of \textbf{FairLP}, an important step is to compute an approximate (unconstrained) $k$-means clustering solution on the dataset by employing the \emph{sklearn} library. 
However, \emph{sklearn} tends to trade accuracy for speed when the dataset gets large. 
As a result, \textbf{FairLP} actually finds a better approximate $k$-means solution on the coreset than on the large dataset \textbf{Census1990} and hence applying coresets can achieve a much smaller clustering objective.

	\section{Conclusion and future work}
	\label{sec:conclusion}
	
	This paper constructs $\eps$-coresets for the fair $k$-median/means clustering problem of size independent on the full dataset, and when the data may have multiple, non-disjoint types.
	Our coreset for fair $k$-median is the first known coreset construction to the best of our knowledge.
	For fair $k$-means, we improve the coreset size of the prior result~\cite{schmidt2018fair}, and extend it to multiple non-disjoint types.
	Our correctness analysis depends on several new geometric observations that may have independent interest.
	The empirical results show that our coresets are indeed much smaller than the full dataset and result in significant reductions in the running time of computing the fair clustering objective.

	Our work leaves several interesting futural directions.
	For unconstrained clustering, there exist several works using the sampling approach such that the coreset size does not depend exponentially on the Euclidean dimension $d$.
	It is interesting to investigate whether sampling approaches can be applied for constructing fair coresets and achieve similar size bound as the unconstrained setting.
	Another interesting direction is to construct coresets for general fair $(k,z)$-clustering beyond $k$-median/means/center.

	\bibliographystyle{plain}
	\bibliography{references}

\appendix
\section{Other Empirical Results}
\label{sec:other_empirical}

In this section, we report the results for a) selecting a binary sensitive type (without normalizing the data), and b) normalizing each dimension to be within $[0,1]$ so that features with large numerical range could not dominate the distance measure.

\subsection{Results: with a binary type}

We choose a binary type for each dataset: sex for \textbf{Adult} and \textbf{Diabetes} and marital for \textbf{Bank}.
The results for the experiments w.r.t. binary types may be found in Tables~\ref{tab:kmedian_single} and~\ref{tab:kmeans_single}.
The observation is that \textbf{Uni} could not achieve smaller empirical errors compared to ours, even for the \textbf{Diabetes} dataset.
A possible explanation is that with more types, the dataset may be better partitioned with respect to types so that \textbf{Uni} performs better compared with the binary type case.

\begin{table}[ht]
	\centering
	\captionsetup{font=small}
	\small
	\caption{performance of $\eps$-coresets for fair $k$-median with one sensitive type w.r.t. varying $\eps$.
	}
	\label{tab:kmedian_single}
	\begin{tabular}{cccccccc}
		\toprule
		& \multirow{2}{*}{$\eps$} & \multicolumn{2}{c}{emp. err.} & \multirow{2}{*}{size} &  \multirow{2}{*}{$T_{S}$ (ms)} & \multirow{2}{*}{$T_{C}$ (ms)} & \multirow{2}{*}{$T_{X}$ (ms)} \\
		& &  Ours & \textbf{Uni} & & & &  \\
		\midrule
		\multirow{7}{*}{\rotatebox{90}{\textbf{Adult}}} &
		10\% & 2.97\% & 32.66\% & 46  & 5 & 363  & 3592 \\
		& 15\% & 3.32\% & 82.47\% & 37  & 5  & 332  & - \\
		& 20\% & 5.39\%  & 30.24\%  & 36  & 5  & 295  & - \\
		& 25\% & 4.44\% & 42.81\% &  28 & 5  & 308  & - \\
		& 30\% & 7.00\% & 30.67\% & 26   & 5  & 304    & - \\
		& 35\% & 6.82\% & 22.46\%  & 30  & 5  & 311    & - \\
		& 40\% & 6.20\%  & 23.55\%  & 24    & 5 & 308  & -  \\
		\midrule
		\multirow{7}{*}{\rotatebox{90}{\textbf{Bank}}} &
		10\% & 1.27\% & 10.08\% & 838  & 21 & 1264  & 2817 \\
		& 15\% & 2.58\% & 4.52\% & 292  & 11  & 652  & - \\
		& 20\% & 3.13\%  & 12.79\%  & 238  & 10 & 607  & - \\
		& 25\% & 3.01\% &16.74\% & 272  & 11  & 605 & - \\
		& 30\% & 4.31\% & 10.93\% & 193   & 9 & 513    & - \\
		& 35\% & 4.80\% & 12.42\% &140   & 7 & 543   & - \\
		& 40\% & 5.56\%  & 12.68\%  & 102    & 7 & 468  & -  \\
		\midrule
		\multirow{7}{*}{\rotatebox{90}{\textbf{Diabetes}}} &
		10\% & 1.23\% & 40.20\% & 51102  & 3766 & 143910  & 14414 \\
		& 15\% & 1.47\% & 14.28\% & 22811  & 909 & 45238  & - \\
		& 20\% & 2.12\%  & 1.84\%  & 7699  & 193  & 15366   & - \\
		& 25\% &2.76\% & 2.57\% & 3159  & 74  & 8402  & - \\
		& 30\% & 3.76\% & 3.47\% & 941   & 23  & 4710   & - \\
		& 35\% & 4.56\% & 4.78\% & 577   & 15 & 4367    & - \\
		& 40\% & 6.33\%  & 10.99\%  & 324    & 11  & 3642  & -  \\
		\bottomrule
	\end{tabular}
	
\end{table}

\begin{table}[ht]
	\centering
	\captionsetup{font=small}
	\small
	\caption{performance of $\eps$-coresets for fair $k$-means with one sensitive type w.r.t. varying $\eps$.
	}
	\label{tab:kmeans_single}
	
	\begin{tabular}{cccccccccc}
		\toprule
		& \multirow{2}{*}{$\eps$} & \multicolumn{3}{c}{emp. err.} & \multirow{2}{*}{size} &  \multirow{2}{*}{$T_{S}$ (ms)} & \multicolumn{2}{c}{$T_{C}$ (ms)} & \multirow{2}{*}{$T_{X}$ (ms)} \\
		& &  Ours  & \textbf{BICO} & \textbf{Uni} & &  & Ours  & \textbf{BICO} &  \\
		\midrule
		\multirow{7}{*}{\rotatebox{90}{\textbf{Adult}}} &
		10\% & 0.91\% & 1.16\% & 184.60\% & 209 & 9 & 547 & 397  & 3908 \\
		& 15\% & 0.78\% & 1.08\% & 30.07\% & 162 & 8 & 468 & 469  & - \\
		& 20\% & 0.584\% & 1.80\% & 63.80\% & 135  & 8  & 516 & 392  & - \\
		& 25\% & 1.40\% & 1.42\% & 33.16\% & 118  & 7  & 540 & 407  & - \\
		& 30\% & 1.58\% & 2.47\% & 52.42\% & 108 & 7  & 521 & 391  & - \\
		& 35\% & 2.29\% & 4.09\% & 100.10\% & 99 & 7 & 534 & 400 & - \\
		& 40\% & 1.79\% & 4.39\% & 90.48\% & 92  & 6  & 510 & 423 & - \\
		\midrule
		\multirow{7}{*}{\rotatebox{90}{\textbf{Bank}}} &
		10\% & 2.87\% & 5.11\% & 19.80\% & 127  & 9 & 1411 & 500  & 2662 \\
		& 15\% & 2.85\% & 5.63\% & 44.81\% & 100 & 7  & 611 & 518  & - \\
		& 20\% & 3.04\% & 4.47\% & 38.91\% & 84 & 6 & 518 & 484  & - \\
		& 25\% & 2.77\% & 6.97\% & 38.27\% & 72  & 7 & 530 & 516  & - \\
		& 30\% & 2.60\% & 6.59\% & 52.24\% & 68  & 6  & 585 & 492  & - \\
		& 35\% & 2.64\% & 8.23\% & 34.05\% & 60  & 6  & 554 & 500  & - \\
		& 40\% & 2.67\% & 8.90\% & 75.58\% & 56  & 6 & 566 & 501  & - \\
		\midrule
		\multirow{7}{*}{\rotatebox{90}{\textbf{Diabetes}}} &
		10\% & 4.44\% & 9.44\% & 3.46\% & 16749& 484 & 65396 & 1035 & 16748 \\
		& 15\% & 8.01\% & 6.88\% & 8.94\% & 1658 & 34  & 11491 & 971 & - \\
		& 20\% & 11.17\% & 14.41\% & 15.87\% & 408 & 11 & 5203 & 872  & - \\
		& 25\% & 15.55\% & 19.05\% & 23.35\% & 158  & 9  & 3047& 922  & - \\
		& 30\% & 14.94\% & 24.62\% & 43.00\% & 104  & 6  & 2849 & 896 & - \\
		& 35\% & 14.72\% & 29.42\% & 16.78\% & 96 & 6 & 2907 & 875  & - \\
		& 40\% & 14.67\% & 25.78\% & 23.26\% & 84  & 7  & 2847 & 875  & - \\
		\bottomrule
	\end{tabular}
\end{table}

\subsection{Results: with normalization}

We choose the same sensitive types for each dataset as in Section~\ref{sec:empirical}, but experiment on the normalized dataset where each feature is normalized to be within $[0,1]$.
The results may be found in Tables~\ref{tab:kmedian_normalize} and~\ref{tab:kmeans_normalize}.
The empirical error and the size of our coreset can be much larger than that without normalization under the same parameter $\eps$.
In particular, for $\eps=10\%$ and the \textbf{Adult} dataset, the empirical error becomes 2x, and the size becomes 20x to that without normalization, for both fair $k$-median and $k$-means.
Moreover, the empirical error of \textbf{Uni} may sometimes be better than our coreset for the \textbf{Diabetes} dataset.
An explanation is that, with normalization, the feature vectors tend to be of a similar norm so they distribute around a hyper-sphere. This makes \textbf{Uni} perfectly suitable for the dataset, as a uniform sampling gives a decent coreset. On the other hand, \textbf{BICO} and our algorithm need to include an $\eps$-net on the sphere, which is of large size.
As a result, uniform sampling could offer superior performance in this case, while our algorithm and \textbf{BICO} can not do better.
Another observation from the tables is that for datasets \textbf{Adult} and \textbf{Bank}, the empirical error of our coreset is smaller than that of \textbf{BICO} when $\eps> 20\%$ but larger when $\eps\leq 20\%$.
A possible explanation is that the way our algorithm works might not capture the pattern of the normalized datasets.
Recall that our algorithm emits rays and project points such that the projection cost is bounded, but we find this part becomes a bottleneck when $\eps$ is small.
Intuitively, if the dataset is well clustered around a few lines, our algorithm should offer superior performance; however, this might not be the case for a dataset that tends to be around a hyper-sphere.

\begin{table}[ht]
	\centering
	\captionsetup{font=small}
	\small
	\caption{performance of $\eps$-coresets for fair $k$-median with normalization w.r.t. varying $\eps$.
	}
	\label{tab:kmedian_normalize}
	\begin{tabular}{cccccccc}
		\toprule
		& \multirow{2}{*}{$\eps$} & \multicolumn{2}{c}{emp. err.} & \multirow{2}{*}{size} &  \multirow{2}{*}{$T_{S}$ (ms)} & \multirow{2}{*}{$T_{C}$ (ms)} & \multirow{2}{*}{$T_{X}$ (ms)} \\
		& &  Ours & \textbf{Uni} & & & &  \\
		\midrule
		\multirow{7}{*}{\rotatebox{90}{\textbf{Adult}}} &
		10\% & 1.08\% & 5.31\% & 22483 & 2085 & 10506 & 7138 \\
		& 15\% & 1.66\% & 3.77\% & 14388  & 1179  & 4835 & - \\
		& 20\% & 2.43\%  & 3.05\%  & 9396  & 643  & 2562   & - \\
		& 25\% & 3.28\% & 1.68\% &  5828 & 361  & 1642  & - \\
		& 30\% & 4.39\% & 3.57\% & 4111   & 244 & 1271  & - \\
		& 35\% & 5.52\% & 2.22\%  & 3409   & 195  & 1125   & - \\
		& 40\% & 6.26\%  & 1.45\%  & 2100   & 113  & 959  & -  \\
		\midrule
		\multirow{7}{*}{\rotatebox{90}{\textbf{Bank}}} &
		10\% & 1.28\% & 11.38\% & 3503  & 165 &1604 & 5286 \\
		& 15\% & 2.90\% & 8.62\% & 1529  & 70  & 938  & - \\
		& 20\% & 4.32\%  & 5.03\%  & 863  & 39  & 696   & - \\
		& 25\% & 7.73\% & 4.13\% &  526  & 25 & 595  & - \\
		& 30\% & 9.52\% & 27.06\% & 329   & 18  & 524   & - \\
		& 35\% & 10.26\% & 15.78\% & 226  & 13 & 486   & - \\
		& 40\% & 9.00\%  & 29.16\%  & 216    & 13 & 492 & -  \\
		\midrule
		\multirow{7}{*}{\rotatebox{90}{\textbf{Diabetes}}} &
		10\% & 0.67\% & 8.29\% & 76380  & 9306  & 96198  & 15303 \\
		& 15\% & 1.15\% & 26.08\% & 55761  & 6321  & 35494  & - \\
		& 20\% & 1.91\%  & 11.04\%  & 34182  & 3560 & 14650  & - \\
		& 25\% & 2.96\% & 2.98\% &  18085 & 1780  & 7076  & - \\
		& 30\% & 4.10\% & 1.84\% & 10834  & 854  & 4112   & - \\
		& 35\% & 5.37\% & 2.12\% & 6402  & 418  & 2712   & - \\
		& 40\% & 6.52\%  & 2.42\%  & 3968   & 234 & 2044 & -  \\
		\midrule
		\multirow{7}{*}{\rotatebox{90}{\textbf{Athlete}}} &
		10\% & 1.25\% & 1.80\% & 3472  & 91  & 5719  & 76081 \\
		& 15\% & 1.99\% & 2.78\% & 1372  & 35  & 3340  & - \\
		& 20\% & 2.94\%  & 5.75\%  & 678 & 20 & 2407  & - \\
		& 25\% & 3.85\% & 7.69\% & 381 & 13  & 1902  & - \\
		& 30\% & 4.57\% & 4.62\% & 208 & 9 & 1582 & - \\
		& 35\% & 6.95\% & 9.22\% & 173 & 9 & 1525 & - \\
		& 40\% & 7.37\%  & 11.96\%  & 117 & 8 & 1406 & -  \\
		\bottomrule
	\end{tabular}
	
\end{table}

\begin{table}[ht]
	\centering
	\captionsetup{font=small}
	\small
	\caption{performance of $\eps$-coresets for fair $k$-means with normalization w.r.t. varying $\eps$.
	}
	\label{tab:kmeans_normalize}
	
	\begin{tabular}{ccccccccccccccc}
		\toprule
		& \multirow{2}{*}{$\eps$} & \multicolumn{3}{c}{emp. err.} & \multirow{2}{*}{size} &  \multirow{2}{*}{$T_{S}$ (ms)} & \multicolumn{2}{c}{$T_{C}$ (ms)} & \multirow{2}{*}{$T_{X}$ (ms)} \\
		& &  Ours  & \textbf{BICO} & \textbf{Uni} & &  & Ours  & \textbf{BICO} &  \\
		\midrule
		\multirow{7}{*}{\rotatebox{90}{\textbf{Adult}}} &
		10\% & 3.39\% & 1.61\% & 3.45\% & 17231  & 1533 &  21038 & 3141 & 17224 \\
		& 15\% & 6.78\% & 3.67\% & 14.14\% & 8876  & 595  & 8178 & 1584  & - \\
		& 20\% & 10.93\% & 7.90\% & 9.38\% & 4087  & 228  & 4047 & 1000  & - \\
		& 25\% & 13.87\% & 12.63\% & 7.12\% & 2213  & 116  & 2652 & 884  & - \\
		& 30\% & 18.36\% & 19.24\% & 4.20\% & 1113  & 56  & 1799 & 844  & - \\
		& 35\% & 19.80\% & 24.98\% & 9.03\% & 791  & 38 & 1237 & 782  & - \\
		& 40\% & 19.64\% & 23.00\% & 3.88\% & 759  & 38 & 1229 & 782  & - \\
		\midrule
		\multirow{7}{*}{\rotatebox{90}{\textbf{Bank}}} &
		10\% & 15.90\% & 11.88\% & 54.84\% & 254  & 15 & 453 & 719  & 4800 \\
		& 15\% & 15.40\% & 16.53\% & 50.86\% & 213  & 13  & 470 & 735 & - \\
		& 20\% & 16.71\% & 16.41\% & 42.25\% & 186& 12  & 470 & 719 & - \\
		& 25\% & 15.25\% & 23.19\% & 42.86\% & 173 & 12  & 448 & 719  & - \\
		& 30\% & 15.23\% & 15.25\% & 44.19\% & 166  & 12  & 455 & 718  & - \\
		& 35\% & 14.99\% & 18.20\% & 54.96\% & 149 & 11& 455 & 719  & - \\
		& 40\% & 14.86\% & 22.81\% & 33.25\% & 149  & 12 & 454 & 765 & - \\
		\midrule
		\multirow{7}{*}{\rotatebox{90}{\textbf{Diabetes}}} &
		10\% & 4.88\% & 9.85\% & 30.26\% & 51841 & 5601 & 111971 & 8891  & 14893 \\
		& 15\% & 8.80\% & 4.65\% & 1.90\% & 20188  & 1798  & 22840 & 3390   & - \\
		& 20\% & 12.72\% & 13.11\% & 2.63\% & 6520  & 372 & 7644 & 1907  & - \\
		& 25\% & 16.96\% & 20.99\% & 3.67\% & 2861  & 144 & 4251 & 1610  & - \\
		& 30\% & 21.03\% & 28.79\% & 6.32\% & 1399  & 72  & 2421 & 1547  & - \\
		& 35\% & 23.23\% & 33.81\% & 8.17\% & 886 & 45  & 1526 & 1521  & - \\
		& 40\% & 27.10\% & 42.67\% & 8.49\% & 478  & 28  & 982 & 1516 & - \\
		\midrule
		\multirow{7}{*}{\rotatebox{90}{\textbf{Athlete}}} &
		10\% & 5.43\% & 3.06\% & 5.50\% & 2350 & 56 & 20560 & 1362  & 72642 \\
		& 15\% & 8.27\% & 9.95\% & 7.32\% & 519 & 15 & 7748 & 1241 & - \\
		& 20\% & 14.18\% & 12.83\% & 23.49\% & 224 & 10 & 3413 & 1216  & - \\
		& 25\% & 15.31\% & 19.06\% & 14.49\% & 127 & 9 & 2059 & 1206  & - \\
		& 30\% & 18.16\% & 22.53\% & 16.78\% & 88 & 7 & 1537 & 1192  & - \\
		& 35\% & 17.95\% & 24.08\% & 26.39\% & 82 & 7  & 1527 & 1224  & - \\
		& 40\% & 18.03\% & 22.41\% & 25.87\% & 74 & 7 & 1539 & 1216 & - \\
		\bottomrule
	\end{tabular}
\end{table}
	
\end{document}